\documentclass[a4paper]{article}
\usepackage{amssymb, amsthm, amsmath}
\usepackage{mathrsfs}
\usepackage{graphicx}
\usepackage{color}
\usepackage{dsfont}
\usepackage{bm}
\usepackage{pgf,tikz}
\usepackage{hyperref}
\usepackage{authblk}
\usepackage{xcolor}
\usetikzlibrary{arrows}
\usepackage{soul} 
\usepackage[authoryear, round]{natbib}

\newtheorem{theorem}{Theorem}[section]
\newtheorem{proposition}[theorem]{Proposition}
\newtheorem{lemma}[theorem]{Lemma}

\theoremstyle{definition}
\newtheorem{definition}[theorem]{Definition}
\newtheorem{assumption}[theorem]{Assumption}
\newtheorem{remark}[theorem]{Remark}
\newtheorem{example}[theorem]{Example}

\newcommand{\R}{\mathbb{R}}

\newcommand{\E}{\mathbb{E}}

\renewcommand{\S}{\mathbb{S}}

\renewcommand{\P}{\mathbb{P}}

\newcommand{\FF}{\mathbb{F}}

\newcommand{\II}{\mathcal{I}}
\newcommand{\Filt}{\mathcal{F}}

\newcommand{\Mset}{\mathcal{M}}

\newcommand{\Tset}{\mathcal{T}}

\newcommand{\Borel}{\mathcal{B}}

\newcommand{\XX}{\mathcal{X}}

\newcommand{\cG}{\mathcal{G}}

\newcommand{\Pset}{\mathcal{P}}
\newcommand{\Qset}{\mathcal{Q}}

\newcommand{\esssup}{\operatornamewithlimits{ess\,sup}}

\let\inf\relax \DeclareMathOperator*\inf{\vphantom{p}inf}

\newcommand\I{\mathds{1}}

\newcommand{\norm}[1]{\left\|#1\right\|}

\newcommand{\as}{\text{a.s.}}
\newcommand{\qs}{\text{q.s.}}

\newcommand{\tr}{\operatorname{tr}}
\newcommand{\Dt}{\mathcal{D}_t}
\newcommand{\Dx}{\nabla_x}
\newcommand{\Dxx}{\nabla_x^2}
\newcommand{\Dth}{\nabla_\theta}

\newcommand{\bcdot}{\boldsymbol{\cdot}}

\title{Robust Pricing and Hedging of American Options in Continuous Time}
\author[1,2]{Ivan Guo}
\author[3]{Jan Ob\l\'oj\thanks{Ivan Guo’s work was partially supported by the Australian Research Council (Grant DP220103106). We are grateful to Anna Aksamit, Libo Li and Gregoire Loeper for fruitful discussions. Our sincere thanks go to Vlad Tuchilus for his help with editing the manuscript.\\ For open access purposes, the authors have applied a CC BY public copyright licence to any author accepted manuscript version arising from this submission.}}

\date{\today}

\affil[1]{School of Mathematics, Monash University}
\affil[2]{Centre for Quantitative Finance and Investment Strategies, Monash University}
\affil[3]{Mathematical Institute, University of Oxford}
\begin{document}
\maketitle

\begin{abstract}
 We consider the robust pricing and hedging of American options in a continuous time setting. We assume asset prices are continuous semimartingales, but we allow for general model uncertainty specification via adapted closed convex constraints on the volatility. We prove the robust pricing-hedging duality. When European options with given prices are available for static trading, we show that duality holds against richer models where these options are traded dynamically. Our proofs rely on probabilistic treatment of randomised stopping times and suitable measure decoupling, and on optimal transport duality. In addition, similarly to the approach of \cite{aksamit2019robust} in discrete time, we identify American options with European options on an enlarged space. 
\bigskip

\end{abstract}

\section{Introduction}
\label{sec:intro}

Option pricing is at the origins of modern quantitative finance, with the Black-Scholes formula one of its most iconic manifestations and often billed as one of the equations which changed the world, see \cite{IanStewart17}. The intuition behind it, going back to the seminal works of \cite{BlackScholes,Merton}, links pricing to hedging: if a hedging strategy can replicate the option's cashflow, then the two should have the same initial price, or else markets would admit plain arbitrage opportunities. In some probabilistic models for the underlying's dynamics, such as the Black-Scholes model, we can compute this price by taking risk-neutral expectations. In more general, incomplete models a perfect hedging strategy may not exist but the principle generalises to the so-called pricing-hedging duality: the cost of the cheapest superhedging strategy is equal to the supremum of the payoff's expectations under all risk-neutral measures, see \cite[Chp.~7]{FollmerSchied2nd} \cite[Chp.~5]{KaratzasShreve1998}.

The above classical results require us to fix a probabilistic model for the dynamics of the risky assets. This perspective, while allowing for an efficient description of the risk within the chosen model, is silent on the errors related to the model choice itself, errors often referred to as the \emph{Knightian uncertainty} after \cite{Knight}, \emph{model ambiguity or misspecification} or \emph{epistemic uncertainty}, see \cite{HansenMarinacci,WalleyImpreciseProbabilities}. 
This shortcoming became a focal point after the 2008 global financial crisis and led to the development of \emph{robust methods} in mathematical finance. Rich literature emerged tackling the foundational questions on how to best capture, mathematically speaking, diversity of models, how to properly define an arbitrage opportunity if a single model is not specified and to obtain a \emph{robust fundamental theorem of asset pricing} that would characterise absence of arbitrage. And how to formulate, and prove, a \emph{robust pricing-hedging duality}. In discrete time settings, these questions were answered using two points of view. One, the so-called \emph{quasi-sure} perspective, relied on families of probability measures which, while potentially mutually singular, have enough structure to obtain a complete set of results, see \cite{BN} and the references therein. The other point of view adopted a pointwise (or pathwise), $\omega$ by $\omega$, also leading to complete robust equivalent of the classical results, see \cite{BFMOZ} and the references therein. The two approaches are naturally just different ways to see the problem and can be shown to be equivalent, see \cite{OWUnified}. 
Analogous efforts were undertaken in continuous time settings, with quasi-sure analysis advanced in, e.g, \cite{DenisMartini,NeufeldNutz2013, PossamaiRoyerTouzi} and pathwise arguments used by, e.g, \cite{Mykland, VovkFS, Allan2023, dolinsky2014martingale, OblojZhou}. The formidable technical challenge in continuous time relates to the necessity to consider suitable classes of admissible trading strategies and to define their outcomes, given by stochastic integrals, simultaneously under many different, often singular, probability measures. To the best of our knowledge, while a lot of progress has been done and many partial exciting results exist, in contrast to the discrete time setting, in continuous time no complete robust theory has been obtained. Our work contributes to the ongoing development of robust mathematical finance in continuous time.

In the context of robust pricing and hedging, if no modelling assumptions are made, the range of possible no-arbitrage prices for a given contingent claim can be impractically large. A natural way to narrow it is to use available market information. In this way, arbitrary constraints coming from modeller's choices are replaced with observable market prices. The latter would be taken into the account in the classical approach via a reverse-engineering process, in which a particular model is selected, say from a parametric class of models, to best fit observed market prices via \emph{calibration}. In the robust approach instead, such market prices are added as constraints to the original problem. This perspective in fact was behind the foundational works of \cite{Hobson, HobsonRogers} which solved robust pricing and hedging of lookback and barrier options using probabilistic methods of Skorokhod embeddings, see \cite{OblojSEP, HobsonSurvey, CoxObloj}, based on the classical observation of \cite{BreedenLitzenberger} that incorporating market prices of call options with a maturity $T$ is equivalent to fixing the marginal distribution at $T$ under the risk-neutral measure. In subsequent seminal contributions \cite{BHLP,GHLT} re-wrote the robust pricing and hedging problem as an optimal transport (OT) problem with a martingale constraint, i.e., a \emph{martingale optimal transport} (MOT) problem. These ideas brought rich and powerful techniques of OT to robust mathematical finance, allowing in particular to develop numerical methods and perform empirical studies for robust methods, see \cite{GuoObloj,EcksteinGuoLimObloj}, and obtain robust pricing-hedging duality results under general constraints, see \cite{GuoTanTouzi, guo2021path}.

Interestingly, as first observed by \cite{Neuberger}, these ideas were not sufficient to consider robust pricing and hedging for American options, and their naive extension could lead to a gap in the pricing-hedging duality. Instead, a weak formulation was suggested by \cite{Neuberger} and further explored in \cite{HobsonNeuberger}. Another technical solution, via randomised models, was proposed by \cite{BayraktarHuangZhou, BayraktarZhou2017}. A more comprehensive explanation was proposed by \cite{aksamit2019robust} who analysed the problem through the lens of information --- while stock prices are observed, and traded on, dynamically, the constraining market option prices are given statically at time $t=0$ and not updated afterwards, limiting the class of exercise policies for the American option holders. They linked duality to a robust version of the dynamic programming principle (DPP), used for classical American option pricing. This also explained why \cite{Dolinsky2014GameOptions}, who studied game options, which include American options, in a nondominated discrete time market, did not encounter duality gap issues since his setup had no statically traded options. \cite{aksamit2019robust} then showed that the pricing-hedging duality can be restored by considering a richer class of models where options are traded dynamically as well. 

Our work makes a decisive new contribution to this literature by developing robust pricing and hedging duality for American options in continuous time. We  show that $\pi^A_{g}$, the superhedging price of an American option $Z$ using dynamic trading in the stocks and static trading in European options with payoffs $g$, is equal to 
$$ \sup_{\widehat\P\in \widehat\Qset, \widehat\tau\in\widehat\Tset^{\widehat\P}} \E^{\widehat\P} Z_{\widehat\tau}, 
$$
where $\widehat \Qset$ corresponds to risk neutral measures for joint dynamics of stocks and options $g$, and the stopping times $\widehat \tau$ are allowed to depend on the information about both. The exact statement of that duality is given in Theorem \ref{thm:americanduality} and it allows to further restrict models to those with volatility taking values in pre-specified sets. These could be general, but could also be given by a point, in which way our results recover the to classical model-specific pricing-hedging duality for American options, as in \cite{Myneni}, as a special case. When no statically traded options are available, $g=0$, then we can restrict to martinagle measures $\Qset$ on $\Omega$ and exercise policies $\tau\in \Tset$. 
To prove the above duality we adopt ideas pioneered in \cite{aksamit2019robust} and show that an American option can be identified with a European option in an enlarged space $\bar\Omega$. However, unlike in \cite{aksamit2019robust}, to link the pricing in the two perspectives, we do not go via a dynamic programming representation but instead rely on probabilistic techniques related to randomised stopping times and Az\'ema supermartingales, see Theorem \ref{thm:americanconvexmtg} and other results in section \ref{sec:robustpricing}. Also on the hedging side, we encounter genuine novel difficulties in the continuous time setup. We lift models to allow for dynamic trading in the European options with given time zero market prices. However, while we work with continuous asset prices, we can not ensure that option prices, defined via conditional expectations, are continuous. In fact, they may be random at time $t=0$. To account for that, we have to consider fictitious markets with trading on $[-\delta, 1]$, see section \ref{sec:robusthedging}. This allows us to establish a chain of inequalities between potential superhedging and pricing valuations, see Proposition \ref{prop:middleinequalities}. The final piece of the jigsaw is to show the LHS and RHS of this long chain of inequalities are equal, which boils down to establishing pricing-hedging duality for European options on $\bar\Omega$. We do this in section \ref{sec:robustPHduality} using OT-duality methods, in analogy to \cite{guo2021path,tan2013optimal}.

\section{Preliminaries }\label{secprelim}


Let $\Omega:=D([0,1];\R^d)$ be the set of right-continuous and with left limits (rcll) paths and $X$ be the canonical process. For each $t\in[0,1]$, let $\Omega_t:=\{\omega_{\cdot\wedge t}:\omega\in\Omega\}$ be the set of paths stopped at time $t$ and  let $\Lambda:= \{(t,\omega): t\in[0,1], \omega\in \Omega_t\}$. 
Let $\FF=(\Filt_t)_{0\leq t\leq 1}$ be the canonical filtration generated by $X$. For any given probability measure $\P$ on $(\Omega,\FF)$, let $\FF^\P$ be the augmentation of $\FF$ with respect to $\P$. We denote by $\Tset$ (resp.\ by $\Tset^\P$) the set of $\FF$ (resp.\ $\FF^\P$) stopping times. 
The spaces $\Omega$ and $\Omega_t$ are equipped the with the norm $\norm{\omega}_\infty=\sup_{t\in[0,1]} |\omega_t|$, while $\Lambda$ is 
equipped with the metric $d_\infty((t,\omega), (t',\omega'))=|t-t'|+\norm{\omega_{\cdot\wedge t}-\omega'_{\cdot\wedge t'}}_\infty$. We note that a process $q$ on $\Omega$ is $\FF$--progressively measurable if and only if it is a measurable map from $\Lambda$ to $\R$.  


Given a Polish space $\XX$ equipped with its Borel $\sigma$-algebra, let $C(\XX)$ be the set of continuous functions on $\XX$, $C_b(\XX)$ be the set of bounded continuous functions and $\Mset(\XX)$ be the set of signed finite Borel measures on $\XX$.
On $C_b(\XX)$, let $\mathfrak{T}_k$ denote the topology of uniform convergence on compact sets of $\XX$. Denote by $\mathfrak{T}_t$ the finest locally convex topology on $C_b(\XX)$ which agrees with $\mathfrak{T}_k$ on closed balls of $C_b(\XX)$ (via the uniform norm). The topology $\mathfrak{T}_t$ was introduced by \cite{lecam1957convergence} and is also known as the ``mixed topology'' \cite{fremlin1972bounded} or the ``substrict topology'' \cite{sentilles1972bounded}. For this paper, we will make use the following key result:
\begin{proposition}[\cite{fremlin1972bounded,sentilles1972bounded}]
The $\mathfrak{T}_t$ dual of $C_b(\XX)$ can be identified with $C_b(\XX)^*=\Mset(\XX)$.
\end{proposition}
\begin{remark}
The choice of topology $\mathfrak{T}_t$ will allow the applications of our duality argument in non-locally compact settings, and to avoid the issue of $C_b(\Omega)^*$ being identified with the set of all regular, signed, finite and finitely additive Borel measures (\cite{dunford1958linear} Theorem IV.6.2) under the usual uniform norm topology.
\end{remark}
Let $\S^d$, $\S^d_+$ and $\S^d_{++}$ denote the sets of symmetric matrices, positive semidefinite matrices and positive definite matrices, respectively, with $a:b := \tr(a^\intercal b)$, for any $a,b\in \S^d$. 
We denote $L^\infty(\XX)$ the set of bounded measurable functions and $L^1(\XX,\mu)$  the set of $\mu$-integrable functions, where $\mu\in\Mset_+(\XX)\subset\Mset(\XX)$ is a positive measure on $\XX$. The respective vector valued version of such spaces are denoted in the natural way, e.g., $C_b(\XX; \R^m)$, $\Mset(\XX; \R^m)$, $L^1(\XX, \mu; \R^m)$ and so on.  In this paper, the typical choices of $\XX$ are $\Omega$, $\Lambda$, $\R^m$, their products, as well as their subspaces. We will also use the shorthand $\mu(f):=\int_\XX f \mu(dx)$. 
By convention, $X$ is a column vector but $q\in C_b(\Lambda; \R^d)$ is a row vector, with $qX$ denoting their scalar product.

Let $\Pset$ be the set of Borel probability measures on $(\Omega,\Filt_1)$. We will work with measures $\P$ under which $X$ is a semimartingale and we will need to define stochastic integrals $\int_0^{\cdot} q_t dX_t$ simultaneously for many such $\P$'s, which may well be mutually singular. Naturally, this requires some restrictions on the integrand $q$ and/or on the measures considered. We defer the details to section \ref{sec:robusthedging}. For now, we introduce the set $\Qset\subseteq\Pset$ of probability measures under which $X$ is a square integrable martingale.

\subsection{The enlarged space}
A key idea in the robust hedging of American options is to progressively encode the optimal stopping decision in the probability space and its filtration. It relies on observing that an American option in the original space can be, in a suitable sense, interpreted as a European option in an enlarged space, see \cite{ElKarouiTan, aksamit2019robust}. In this subsection, we will set up the appropriate enlarged space and establish the robust pricing-hedging duality for European options in this space.
Consider the set $\Theta\subset C([0,1],\R)$ defined by 
\[
\Theta:= \{\vartheta \in C([0,1],\R): \vartheta_t= \theta\wedge t, \text{ for some } \theta\in[0,1]\}.
\]
The space $\Theta$ is isometric to $[0,1]$, with an obvious bijection between $\theta\in[0,1]$ and $\vartheta\in\Theta$ via $\vartheta=\theta\wedge\cdot$ or $\theta=\vartheta_1$. 
When there's no confusion, we will simply use $\theta$ to denote elements of $\Theta$, while $\vartheta$ is reserved when it is necessary to refer to the whole path.
Define the product space $\bar\Omega:=\Theta \times \Omega$, which is a subset of $D([0,1],\R^{d+1})$ and inherits its topology. Also extend the canonical process via $\bar X_t(\bar\omega)=(\vartheta_t, X_t(\omega))$ and the canonical filtration to $\bar \FF$.
For each $\bar\omega=(\vartheta,\omega)\in\bar\Omega$, let $\bar\omega_{\cdot\wedge t}:=(\vartheta_{\cdot\wedge t},\omega_{\cdot\wedge t})$ denote the stopped path at time $t$. 
Analogous to $\Lambda$ from the original space, we define $\bar\Lambda := \{(t,\bar\omega_{\cdot\wedge t}): t\in[0,1], \bar\omega\in \bar\Omega\}$. As before, for convenience, we will often write elements of $\bar\Lambda$ as
$(t,\bar\omega_{\cdot\wedge t}) = (t,\vartheta_{\cdot\wedge t},\omega_{\cdot\wedge t})=(t,\theta\wedge t,\omega_{\cdot\wedge t})$.

Similar to before, let $\bar\Pset$ be the set of Borel probability measures on $\bar\Omega$ and let $\bar\Qset \subset \bar\Pset$ be the set of measures under which $X$ is a square integrable martingale.

\section{Robust pricing of American options}
\label{sec:robustpricing}

In this section, we represent the robust price of an American option as the robust price of a corresponding European option in the enlarged space. Put differently, we show that optimising linear functionals on $\Tset\times \Qset$ is equivalent to optimising on the enlarged space $\bar\Qset$. 
An American option payoff is given as a measurable function $Z\in L^0(\Lambda)$, i.e., if exercised at time $t$, the option pays $Z_t(\omega) = Z(t,\omega_{t\land \cdot})$, and such an option can be naturally also exercised at a stopping time $\tau$ and we write $Z_\tau$ for its payoff $Z(\tau(\omega),\omega_{\cdot \land \tau(\omega)})$. Note also that $Z$ naturally induces a random variable on $\bar\Omega$ via $\bar Z(\bar\omega) = Z(\theta,\omega_{\cdot \land \theta})$, i.e., the composition of $Z$ with the map 
$$\bar\Omega\ni \bar\omega=(\vartheta,\omega) \longrightarrow (\vartheta_1,\omega_{\cdot \land \vartheta_1})\in \Lambda,
$$
which is continuous with our choice of norms. 

For any $E=(E_1,E_2) \in \Borel(\Omega) \times \Borel(\Lambda)$,
define $\bar E_1 := \Theta \times E_1 \in \Borel(\bar\Omega)$ and 
\[
\bar E_2 = \{(t,\theta\wedge t, X_{\cdot\wedge t})\in \bar\Lambda : (t, X_{\cdot\wedge t})\in E_2\} \in \Borel(\bar\Lambda).
\]
Next, define
\begin{equation}
    \begin{split}\label{eq:QsetEdef}
        \Qset^E&=\{\P\in\Qset: \P(E_1)= \lambda\otimes\P (E_2)=1\},\\
\bar\Qset^E&=\{\bar\P\in\bar\Qset: \bar\P(\bar E_1)= \lambda \otimes \bar\P (\bar E_2)=1\},
    \end{split}
\end{equation}
where $\lambda$ is the Lebesgue measure, so $\lambda\otimes\P\in \Pset(\Lambda)$ and $\lambda \otimes \bar\P \in \Pset(\bar\Lambda)$.

We state now our main result:
\begin{theorem}\label{thm:americanconvexmtg} 
Let $E_1\in \Borel(\Omega)$ and $ E_2\in \Borel(\Lambda)$.
Then for any $Z\in L^0(\Lambda)$ and bounded from below, we have
\[
\sup_{ \P\in\Qset^E, \tau\in\Tset^\P} \E^\P Z_\tau = \sup_{\bar\P\in\bar\Qset^E} \E^{\bar\P} Z(\theta,\omega_{\cdot\land \theta}).
\]
\end{theorem}
The set $E$ above is a technical device that will allow us to apply Theorem \ref{thm:americanconvexmtg} to a variety of settings. 
The proof of the theorem is split into two main steps. Lemma \ref{lem:americanconvex} connects measures on $\bar\Omega$ to measures on $\Omega$ and stopping times. Lemma \ref{lem:americanconvexmtg} ensures that the martingale property is preserved along the way. The latter, we feel, is the more subtle of the two steps. If one does not care about the preservation of the martingale property, then there are much more direct approaches to the first step.

\subsection{Convexifying stopping times and measures}


We are interested in pairs of the form $(\tau,\P)$ where $\P$ belongs to some subset of probability measures on $\Omega$ and $\tau$ is an $\FF^\P$-stopping time, $\tau\in \Tset^\P$.
Each pair $(\tau,\P)$ defines a linear map on the space of jointly measurable functions on $(t,\Omega)$ via $\phi\to \E^\P(\phi_\tau)$.
In order to employ convex duality techniques to solve optimisation problems of the form $\sup_{(\tau,\P)}\E^\P(\phi_\tau)$, we need to identify the convex hull of such pairs $(\tau,\P)$ or linear maps. 
Since many problems have objectives or constraints on $\P$, the convexification should be done in a way that preserves the convexity of $\P$ (or the induced linear map on measurable functions on $\Omega$).

A natural approach (see e.g., \cite{ElKarouiTan, aksamit2019robust}) is to identify $(\tau,\P)$ with a measure from $\Pset(\bar\Omega)$. However, a general element in $\Pset(\bar\Omega)$ usually corresponds to a random time (not even a randomised stopping time, see Example \ref{examplerandomtime}). The approach of \cite{aksamit2019robust} overcomes this by explicitly constructing an optimal stopping time via backward induction in discrete time, which cannot be applied in continuous time settings.
Our approach is to use $\Pset(\bar\Omega)$ as well and to use the non-anticipative nature of American payoffs together with the martingale property of our pricing measures to show that we can indeed restrict to randomised stopping times.

\begin{definition}\label{def:randST}
    A \emph{randomised stopping time} is an $(\Omega,\FF)$-adapted, right-continuous, increasing process $A$ with $A_0=0$ and $A_1=1$.
\end{definition}
Randomised stopping times have been studied in 
\cite{BaxterChacon}, see also \cite{MeyerBaxterChacon}, and in many works since, see \cite{BCHInventiones2017} for a good discussion in relation to OT methods for Skorokhod embeddings. If $\tau$ is an $\FF$-stopping time, then $A^\tau_t=\I_{\tau\leq t}$ is a randomised stopping time which has a particularly simple structure: it is equal to zero and then jumps to one at time $\tau$. More generally, $A_t$ can be interpreted as the proportion of mass that has been stopped along each path by the time $t$. The non-anticipative property of $A$ helps eliminate the possibility of random times that are not randomised stopping times. Our aim is to extract an appropriate non-anticipative $A$ from each $\mu\in\Pset(\bar\Omega)$. First however, we give a simple example to show that naive convexification of pairs $(\tau,\P)$ would not result in randomised stopping times. 

\begin{example}\label{examplerandomtime}
    Consider two pairs of stopping times and martingale measures, $(\tau,\P)$ and $(\tau',\P')$, where $\tau=0$, $\tau'=1$, $\P$ is the law of a constant process, and $\P'$ is the law of the same constant process on $t\in[0,1/2]$ followed by a Brownian motion on $t\in[1/2,1]$. Let $\bar\P = \delta_0 \otimes \P$ and $\bar\P' = \delta_1 \otimes \P'$ be the corresponding measures in the enlarged space. The average $\bar\P''=(\bar\P+\bar\P')/2$ is a martingale measure, however its $\theta$ component does not correspond to a stopping time or a randomised stopping time, as it requires the knowledge of the process from $t=1/2$ to determine the stopping decision at $t=0$. In particular, at $t=0$ it will immediately stop all paths that will be constant on $t\in[1/2,1]$. However, at time $t=0$, it is not possible to distinguish those paths from those which will follow a Brownian motion on $[1/2,1]$. 
\end{example} 
 
 We recall the following simple properties of randomised stopping times. 

\begin{lemma}[\cite{shmaya2014equivalence}]\label{randomisedstopping}
Let $A$ be a randomised stopping time and for each $r\in[0,1]$, define $\tau_r=\inf\{t \geq 0 :A_t\geq r\}$. Then $(\tau_r, r\in [0,1])$ is a non-decreasing family of stopping times and 
\[
A_t(\omega) = \lambda ({r\in [0,1]: \tau_r(\omega)\leq t}),\quad t\in [0,1],
\]
where $\lambda$ is the Lebesgue measure. Consequently, for any $\eta\in L^\infty(\bar\Omega)$,
\[
\int_0^1 \eta(\theta,\omega)\, dA_\theta = \int_0^1 \eta(\tau_r,\omega)\, dr.
\]
\end{lemma}

We also need to recall the It\^o--Watanabe multiplicative decomposition of non-negative supermartingales. 
\begin{theorem}[\cite{ito1965transformation}]\label{thmitowatanabe}
Let $\xi$ be a non-negative right continuous supermartingale with $\xi_0>0$ defined on a filtered probability space satisfying the usual hypothesis. Then $\xi$ has the decomposition \[\xi=MA\] with a positive local martingale $M$ and a natural decreasing process $A$. If there are two such factorisations, then they are identical up to $T_\xi=\inf\{t \geq 0: \xi_t=0\}$.

Furthermore, if there exists a constant $K>0$ and an almost surely finite stopping time $T$ such that $1/K\leq \xi_t\leq K$ for $t<T$ and $\xi_t=\xi_T$ for $t\geq T$, then the local martingale $M$ in the decomposition is a (true) martingale.
\end{theorem}

The distinction between local and true martingale in Theorem \ref{thmitowatanabe} actually creates some difficulties for us, but we will work around it by focusing on the following type of measures in $\Pset(\bar\Omega)$. 
\begin{definition}
    For $\epsilon\in (0,1)$, we say that $\mu\in \Pset(\bar\Omega)$ is \emph{$\epsilon$-modified} if $\mu(\{1\}\times \Gamma) \geq \epsilon \mu^\Omega(\Gamma)$ for any $\Gamma\in \Filt_1$, and where $\mu^\Omega$ is the $\Omega$-marginal of $\mu$. 
\end{definition}
\begin{lemma}
    Let $\epsilon\in (0,1)$, $\mu\in \Pset(\bar\Omega)$ and $\delta_1\in\Pset([0,1])$ be the Dirac mass at 1. Then $\mu^\epsilon := (1-\epsilon) \mu + \epsilon (\delta_1\times \mu^\Omega)$ is $\epsilon$-modified and has the same $\Omega$-marginal as $\mu$. In addition, if $\mu\in \bar \Qset^E$ then $\mu^\epsilon \in \bar \Qset^E$.
\end{lemma}
The proof is immediate and we omit the details. We call $\mu^\epsilon$ the $\epsilon$-modification of $\mu$. In our proofs, we will establish the desired statements for $\mu^\epsilon$ and take $\epsilon \searrow 0$. 

The following lemma is the first step in the proof of Theorem \ref{thm:americanconvexmtg}. It shows that when the American payoff is adapted we can identify a measure on $\bar\Omega$ with a measure on $\Omega$ and a randomised stopping time. This randomised stopping time is extracted by applying the multiplicative decomposition to an Az\'ema supermatingale (\cite{Azema1972Inventiones}), a technique also used in \cite{li2022vulnerable} in the context of option pricing with default risks.  In the subsequent section we will look at the martingale property of such measures. 
\begin{lemma}
    \label{lem:americanconvex}
Let $\mu\in \Pset(\bar\Omega)$ be $\epsilon$-modified. Then there exists a probability measure $\P\in\Pset(\Omega)$ equivalent to $\mu^\Omega$ and a right-continuous, increasing and $\FF^\P$-adapted process $A$ with $A_0=0$ and $A_1=1$, such that for every $\FF^\P$-adapted process $\psi$,
\[
\mu( \psi_\theta(\omega) ) = \E^{\P} \int_0^1 \psi_t \, dA_t .
\]
This can be written as 
\[
\mu( \psi_\theta(\omega) ) = \E^{\P} \int_0^1 \psi_{\tau_r} \, dr.
\]
for a family of $\FF^\P$-stopping times $\{\tau_r\}_r$ that satisfies $\tau_r=1$ for $r>1-\epsilon$.

\end{lemma}

\begin{proof} Instead of $\psi$, it suffices to check functions of the form $\psi_t=\I(t> s, \omega\in \Gamma)$ for any fixed $s$ and any $\Gamma\in \Filt^{\mu_\Omega}_s$.

Define the raw IV process $R$ from $\mu$ (by say, disintegrating in $\omega$ and taking the distribution functions with respect to $\theta$), so
\[
\int_{\bar\Omega} \I(\theta> s, \omega\in \Gamma) \, d\mu = \E^{\mu_\Omega} \big( (1-R_s)\I(\omega\in \Gamma) \big),\quad s\in [0,1].
\]
Let ${}^o R_t:=\E^{\mu_\Omega}(R_t\,|\, \Filt^{\mu_\Omega}_{t})$ be the optional projection of $R$. Then $1-{}^o R$ is a non-negative $\mu_\Omega$-supermartingale, and in particular we can and will work with its rcll version, 
often known as the Az\'ema supermartingale after \cite{Azema1972Inventiones}. 
Moreover, since $\mu$ is $\epsilon$-modified, $1\geq 1-{}^o R_1>\epsilon$.
By Theorem \ref{thmitowatanabe}, it admits a multiplicative decomposition $1-{}^o R=M(1-A)$, for some positive $(\FF^{\mu_\Omega},\mu_\Omega)$-martingale $M$ and an increasing $\FF^{\mu_\Omega}$-adapted process $A$ with $M_0=1, A_0=0, A_1=1$. Therefore,
\begin{align*}
\E^{\mu_\Omega} \big( (1-R_s)\I(\omega\in \Gamma) \big) &=\E^{\mu_\Omega}  \big( (1-{}^o R_s)\I(\omega\in \Gamma) \big)\\
&=\E^{\mu_\Omega} \big( M_s(1-A_s)\I(\omega\in \Gamma) \big)\\
&=\E^{\P} \big((1-A_s)\I(\omega\in \Gamma) \big),
\end{align*}
where $d\P/d\mu_\Omega =M_1$. 
The last part follows from Lemma \ref{randomisedstopping}.
\end{proof}

\begin{remark} The above proposition has the following interpretation. 
We showed that every measure in $\Pset(\bar\Omega)$ is equivalent to some measure in $\Pset(\Lambda)$, when tested only against non-anticipative functions. The measure from $\Pset(\Lambda)$ is then equivalent to a pair $(A,\P)$ where $A$ characterises a randomised stopping time, while $\P$ is only uniquely determined until the randomised stopping time finishes. This makes sense since elements of $\Lambda$ cannot see into the future.
\end{remark}

\subsection{Preservation of the martingale property}

Now we will see that if $\mu$ was a martingale measure in $\Pset(\bar\Omega)$, then $\P$ in Lemma \ref{lem:americanconvex} is also a martingale measure in $\Pset(\Omega)$ with the same characteristics on $\Omega$.

\begin{lemma}\label{lem:americanconvexmtg}
Let $\mu\in\bar \Qset$ be $\epsilon$-modified.
Then $\P$ from Lemma \ref{lem:americanconvex} satisfies $\P\in \Qset$. If in addition $\mu\in \bar\Qset^E$, then $\P\in \Qset^E$.
\end{lemma}

\begin{proof}
Take $\P$ and $\{\tau_r\}_{r\in [0,1]}$ from Lemma \ref{lem:americanconvex}. Since $\tau_r:\Omega\to \R$ is an $\FF^\P$-stopping time (and hence $\FF^{\mu_\Omega}$-stopping time), its (obvious) embedding as a map $\bar\Omega\to\R$ is an $\bar\FF^\mu$-stopping time. This follows from the fact that if $\Gamma$ is a null set under $\mu_\Omega$, then $\Theta \times \Gamma$ is a null set under $\mu$.

For any $s, a\in[0,1]$ and $\Gamma\in\Filt^\P_s$, we let $\II_t(\omega):=\I_\Gamma(\omega)\I_{(s,1]}(t)(X_{t}-X_s)$ and note that $t\to \II_{t\wedge \tau_a}$ is an $\FF^\P$-adapted process. We will show that 
\begin{align}\label{eqconvexmtg100}
\int_a^1 \E^{\P} (\I_\Gamma\I_{\tau_r>s}(X_{\tau_r}-X_s)) \, dr = \int_a^1 \E^\P( \II_{\tau_r} ) \, dr = 0.
\end{align}
Recall that since $\mu$ is $\epsilon$-modified, we have $\tau_a=1$ for any $1>a>1-\epsilon$. Choosing such an $a$ shows that for each $s<1$, $\E^{\P} (\I_\Gamma(X_1-X_s))=0$, thus $X$ is an $(\FF^\P,\P)$-martingale. 

It remains to establish \eqref{eqconvexmtg100}. By the optional stopping theorem
\[
\mu(\II_{\theta\wedge\tau_a} )=0.
\]
This is because both $\theta$ and $\tau_a$ are $\bar\Filt^\mu$-stopping times, $\I_\Gamma\I(\theta\wedge\tau_a>s)$ is $\bar\Filt^\mu_s$ measurable, while $X$ is an $(\bar\FF, \mu)$- and hence $(\bar\FF^\mu, \mu)$-martingale.
Applying Lemmas \ref{randomisedstopping}, \ref{lem:americanconvex} and Fubini's theorem,
\begin{align*}
0=\mu( \II_{\theta\wedge\tau_a} ) &= \E^\P \int_0^1  \II_{t\wedge\tau_a} \,dA_t = \int_0^1 \E^\P( \II_{\tau_r\wedge\tau_a} )\,dr \\
&= \int_0^a \E^\P( \II_{\tau_r} )\,dr + (1-a) \E^\P( \II_{\tau_a}).
\end{align*}
Let $C_r= \E^\P( \II_{\tau_r} )$, we arrive at the equation
\[
\frac{1}{a-1}\int_0^a C_r\, dr = C_a.
\]
Hence $C_a$ is differentiable for $a<1$. Multiplying both sides by $(a-1)$ and differentiating, 
\[
C_a = C_a + (a-1) \partial_a C_a \quad \implies \quad \partial_a C_a =0,\ a<1.
\]
The only possible solution is $C_r=0$ for $r<1$.
Therefore our claim \eqref{eqconvexmtg100} is proven.

Finally, if $\mu\in \bar\Qset^E$, then by definition $\mu_\Omega\in\Qset^E$. Since $\P$ is absolutely continuous with respect to $\mu_\Omega$, we also have $\P\in\Qset^E$.
\end{proof}

\begin{proof}[Proof of Theorem \ref{thm:americanconvexmtg}]
Given each $\P\in\Qset^E$ and $\tau\in\Tset^\P$, we can define $\bar\P \in \Mset(\bar\Omega)$ to be the pushforward measure of $\P$ with respect to the map $\omega \to (\tau(\omega), \omega)$. It is straightforward to check that $\bar\P\in\bar\Qset^E$ and hence 
\[
\sup_{ \P\in\Qset^E, \tau\in\Tset^\P} \E^\P Z(\tau,\omega) \leq \sup_{\bar\P\in\bar\Qset^E} \E^{\bar\P} Z(\theta,\omega).
\] 
For the reverse direction, we first show that it suffices to only focus on $\epsilon$-modified measures. Note that  $\bar\P\in\bar\Qset^E$ implies that $\bar\P^\epsilon\in\bar\Qset^E$.
Since $Z$ is bounded below, there exists some constant $C$ such that $\E^{\bar\P^\epsilon} Z \geq (1-\epsilon) \E^{\bar\P} Z + \epsilon C$. So if we can prove
\[
\E^{\bar\P^\epsilon} Z  \leq \sup_{\P\in\Qset^E, \tau\in\Tset^\P} \E^\P Z(\tau,\omega),
\]
for every such $\bar \P^\epsilon$, the we can complete the proof by taking $\epsilon\to 0$. 

So, without loss of generality, assume that $\bar\P\in \bar\Qset^E$ is already $\epsilon$-modified. By Lemma \ref{lem:americanconvex}, to $\bar\P$ we can associate $\P\in \Pset$ such that 
\[
\E^{\bar\P} Z(\theta,\omega) =  \int_0^1 \E^{\P} Z(\tau_r,\omega) \, dr.
\]
Applying Lemma \ref{lem:americanconvexmtg}, we deduce that $\P\in \Qset^E$. 
We conclude that
\[
\E^{\bar\P} Z(\theta,\omega) \leq \esssup_{r\in[0,1]} \E^{\P} Z(\tau_r,\omega) \leq \sup_{\P\in\Qset^E, \tau\in\Tset^\P} \E^\P Z(\tau,\omega).
\]
Since this holds for every $Z$, our proof is complete. Note that we have established, in particular, that $\Qset^E=\emptyset$ if and only if $\bar\Qset^E=\emptyset$.
\end{proof}

\begin{remark}
Fix a random variable $g:\Omega \to \R$. Define $\Qset^E_g := \Qset^E \cap \{\P:\E^\P (g)=0\}$ and $\bar\Qset^E_g := \bar\Qset^E \cap \{\bar\P:\E^{\bar\P} (g)=0\}$.
We cannot apply Theorem \ref{thm:americanconvexmtg} to the set $\Qset^E_g$. In fact, we may have
\[
\sup_{ \P\in\Qset^E_g, \tau\in\Tset^\P} \E^\P Z(\tau,\omega) < \sup_{\bar\P\in\bar\Qset^E_g} \E^{\bar\P} Z(\theta,\omega).
\]
The inequality may arise because when we construct the measure $\P$ from $\bar\P_\Omega$, there is no guarantee that the constraint $\E^\P g =0$ is preserved. This is further explored in section \ref{sectiondualitygap}, which shows that when static European options are included, there could be a duality gap between the robust hedging and robust pricing. Both will follow from Example \ref{ex:dualitygap} and Theorem \ref{thm:americanduality}. 
\end{remark}

\section{Robust hedging of American options and weak pricing-hedging duality}
\label{sec:robusthedging}

\subsection{Pathwise stochastic integration}
Dynamic trading is described through stochastic integrals, so we need to define $\int_0^{\cdot} q_t dX_t$ simultaneously for many martingale measures $\P$'s, which may well be mutually singular. This requires some restrictions on the objects we consider. In our setup, we can naturally restrict to suitably continuous integrands $q$ and  define the stochastic integrals pathwise as limits of certain discrete Riemann-sum approximations, setting the integrals to be zero if limits are not defined. A successful construction simply has to ensure that the set of paths for which the limits exist has full $\P$-measure, for all $\P$'s considered, and the pathwise integrals coincide $\P$-a.s. with their classical It\^o counterparts. Such an approach to pathwise integration has been considered in many works; see \cite{Karandikar, Follmer, Peng, DenisMartini, DavisOblojSiorpaes, Promeletal, ContFournier}. It is worth noting that a different approach to this question of \emph{aggregation} was proposed in \cite{SonerTouziZhang} who worked with the properties of the set of measures considered instead. 

For our purposes, \cite{Karandikar} provides the most elegant and simplest aggregation method. Given two functions $\rho,\eta$ on $[0,t]$, $\rho$ right-continuous and with left limits (rcll), or left-continuous with right limits (lcrl) and $\eta$ rcll, their integral $\rho\bcdot \eta$ is defined as a rcll function on $[0,t]$, and $(\rho\bcdot \eta)_t = (\tilde \rho\bcdot \eta)_t$ if $\rho$ and $\tilde \rho$ are versions of each other with different (right- or left-) continuity properties. In fact, $\rho\bcdot \eta$ is given as the limit, in the topology of uniform convergence, of $\rho^n\bcdot \eta$, where $\rho^n$ is piece-wise constant and hence the integral is simply a sum, with $\rho\bcdot \eta=0$ on the set where the sequence does not converge. The latter is the same as the set where the sequence is not Cauchy and is a measurable set in $\Filt_t$. It is also a null set under any semimartingale measure, since by \cite[Thm.~3]{Karandikar} this pathwise construction a.s.\ coincides with the It\^o stochastic integral of an adapted rcll (or lcrl) process against a continuous semimartingale, both defined on some filtered probability space satisfying the usual hypothesis. We also note that if we consider the construction on $[0,s]$ and $[0,t]$, with $s<t$, then the two approximations $\rho^n\bcdot \eta$ coincide on $[0,s]$. It follows that $(\rho \bcdot \eta)_{t\in [0,1]}$ is a measurable map on $\Lambda$, i.e., is progressively measurable and has continuous paths on the ($\Filt_1$--measurable) set where the approximations converge uniformly on $[0,1]$. This justifies our use of the integral notation, we write:  
\begin{equation}\label{eq:defSI}
(q(\omega)\bcdot X(\omega))_t = \int_0^t q_s(\omega)dX_s(\omega) = \left(\int_0^t q_s dX_s\right)(\omega),
\end{equation} 
where usually $q\in C(\Lambda)$ and the process $q_t(\omega) := q(t,\omega_{\cdot \land t})$, $t\in [0,1]$ has continuous paths for all $\omega\in \Omega$ and is $\FF$-progressively measurable. For $q\in C(\Lambda;\R^d)$ and $X$ an $\R^d$ valued process, the integrals are defined component-wise.  

\subsection{Robust superhedging of American Options}

We return now to the problem of superhedging an American option with payoff $Z\in C_b(\Lambda)$. The seller of the option wants to robustly hedge their exposure. They can trade dynamically in the underlying $X$ and, since they observe the exercise time, can also adjust the hedging strategy accordingly: we start with a hedging strategy $q$ and switch to $q^u$ if the option is exercised at time $u$. To parametrise the latter, we introduce the set 
$$\bar\Lambda_{\geq} = \{(t,u,\omega_{t\land \cdot}): u\in [0,1], t\in [u,1], \omega\in \Omega \}\subseteq \bar \Lambda,$$
with metric inherited from $\bar \Lambda$. We recall the sets of martingale measures $\Qset^E$ and $\bar\Qset^E$ defined in \eqref{eq:QsetEdef}. We will later use $E_2$ to encode the seller's pathwise beliefs about the range of asset behaviour they want to hedge against. In addition, the seller can trade statically in the European options at their market prices observed today. Without any loss of generality, we can assume these options have zero prices today since, as we ignore transaction costs, we can simply shift the options' payoffs by their true prices. The options' payoffs are thus given by some $g\in C_b(\Omega; \R^m)$, where we assume $g$ is bounded to be able to later apply OT-duality methods, see Proposition \ref{prop:Euroduality}. 
Naturally, the motivating example is that of $g$ being a vector of put payoffs, call options being converted to puts via call-put parity. Finally, the set of calibrated martingale measures is given by $\Qset^E_g=\{\P\in \Qset^E: \E^\P(g)=0\}$. We recall also that for $\P\in \Qset$, $\FF^\P$ is the right-continuous and completed version of the natural filtration. Throughout this section we assume that $\Qset^E$, $\Qset^E_g$ are non-empty. 

The superhedging price of the American option is defined as follows:
\begin{align*}
\pi^A_{g,E}(Z):=\inf\bigg\{ & x: \exists (q,\tilde q,h)\in C_b(\Lambda;\R^d)\times C_b(\bar\Lambda_{\geq};\R^d) \times \R^m,\ \text{s.t.}\  x+\int_0^u q(t,\omega_{t\land \cdot})dX_t(\omega)\\
 & \quad +\int_u^1 \tilde q^u(t,\omega_{t\land \cdot}) dX_t(\omega) + h g \geq Z(u,\omega_{\cdot\wedge u}),\ \forall\, u\in[0,1],\  \Qset^E_g(d\omega)\text{-\qs} \bigg\},
\end{align*}
where $\Qset^E_g\text{-\qs}$ means $\P$-a.s.\ for any $\P\in\Qset^E_g$. 
When we want to stress the dependence on the set restrictions on the martingale test measures, we write $\pi^A_{g,E}(Z)$. 
On the other hand, when there are no traded European options, $g=0$, we may simply write $\pi^A(Z)$. 
It is intuitive that in this case, there is no need to hedge after the exercise time.
\begin{lemma} \label{lem:simpleAhedge}
Let $Z\in C_b(\Lambda)$ and suppose $g=0$. Then $\pi^A_E(Z)=\tilde \pi^A_E(Z)$, where 
$$ \tilde \pi^A_E(Z):=\inf\bigg\{  x: \exists q\in C_b(\Lambda;\R^d)\ \text{s.t.}\  x+\int_0^u q_tdX_t \geq Z_u,\ \forall\, u\in[0,1],\  \Qset^E(d\omega)\text{-\qs}, \bigg\}.
$$
\end{lemma}
\begin{proof}
    It is clear that $\tilde \pi^A(Z)\geq \pi^A(Z)$. Let $(q,\tilde q,h)$ be a superhedging strategy for $\pi^A(Z)$ starting from $x$. We claim that $(x,q)$ is then a superhedging strategy for $\tilde \pi^A(Z)$ which gives the reverse inequality.  
    Fix $u\in [0,1]$ and $\P\in \Qset^E$. Let $M^u_t = \int_u^t \tilde q^u_t dX_t$, $t\in [u, 1]$, which is a $\P$--martingale with $M^u_u=0$. The superhedging property means that 
$$ x+\int_0^u q_t dX_t +  M^u_1 \geq Z_u, \quad \P-\text{a.s.}
$$
and taking $\Filt_u$-conditional expectations, we obtain $x+\int_0^u q_t dX_t \geq Z_u$ $\P$-a.s. Since both LHS and RHS are right-continuous processes, we see that the inequality holds for all $u\in [0,1]$, $\P$-a.s. and hence also $\Qset^E$-q.s. since $\P$ was an arbitrary element of $\Qset^E$.  
\end{proof}
From the above proof it follows that for any $u\in [0,1]$ and any $\P\in \Qset^E$, $M^u$ is a martingale starting in zero and with $M^u_1 \geq 0$ $\P$-a.s., and hence in fact $M^u_1=0$ $\Qset^E$-q.s., so that we not only can, but must, have $\tilde q\equiv 0$ $\Qset^E$-q.s.

We now come back to the richer case when $g\neq 0$ and observe that the so-called weak pricing-hedging duality holds in general. 
\begin{proposition}\label{prop:Americanweakduality}
    Let $Z\in C_b(\Lambda)$, $g\in C_b(\Omega; \R^m)$ and suppose $\Qset^E_g\neq \emptyset$. Then 
    $$\pi^A_{g,E}(Z) \geq \sup_{\P\in \Qset^E_g,\tau\in\Tset^{\P}} \E^{\P} Z_\tau.$$
\end{proposition}
\begin{proof}
Take any $\P\in \Qset^E_g$, $\tau\in\Tset^\P$ and any $\pi^A_g(Z)$ superhedging strategy $(x,q,\tilde q,h)$. For $t\in [0,1]$, note that $\omega \to \tilde q^{\tau(\omega)}(t,\omega_{t\land \cdot})\I_{t>\tau(\omega)}$ is a composition of $\omega \to (t, t\land \tau(\omega),\omega_{t\land \cdot})$, which is $\Filt_t^\P$-measurable with $(t,u,\omega_{t\land \cdot}) \to \tilde q^u(t,\omega_{t\land \cdot})\I_{t > u}$ which is measurable. It follows that $t\to \tilde q^\tau(t,\omega_{t\land \cdot})\I_{t>\tau}$ is adapted and left-continuous (with right limits), so progressive and 
$$\int_{\tau}^1 \tilde q^\tau dX_t = \int_0^1 \tilde q^{\tau(\omega)}(t,\omega_{t\land \cdot})\I_{t>\tau}d X_t$$
is well defined $\P$-a.s., and it coincides $\P$-a.s. with the pathwise version in \eqref{eq:defSI}. Naturally, the pathwise stochastic integral of $\tilde q^\tau \I_{t>\tau}$ and of $\tilde q^u \I_{t>u}$ coincide for $\omega\in\Omega$ such that $\tau(\omega)=u$. By the superhedging property we thus have
$$x + \int_0^1 q\I_{t\leq \tau}dX_t + \int_0^1 \tilde q^\tau \I_{t>\tau} dX_t + hg \geq Z_\tau$$
$\Qset_g^E$-q.s.\ and in particular $\P$-a.s., and hence the same equality holds $\P$-a.s., with the stochastic integral replacing the pathwise integral. The desired inequality is obtained taking expectations under $\P$ on both sides, noting that the expectation of the stochastic integral is zero since $\tilde q$ is bounded and $\P\in \Qset^E_g$, so $X$ is a square integrable martingale under $\P$. 
\end{proof}

\subsection{Duality gap and a dynamic extension for statically traded European options}\label{sectiondualitygap}

As discussed in the introduction, the pricing-hedging duality may then fail for an American option, i.e., the inequality in Proposition \ref{prop:Americanweakduality} may be strict. This was originally observed by \cite{Neuberger, HobsonNeuberger}, in a discrete time setting, and explored in a series of subsequent works. In particular, \cite{aksamit2019robust} linked this to the failure of the dynamic programming principle and shown that if the option exercise times are allowed to depend on richer information, namely on dynamic prices of $g$, then pricing-hedging duality is recovered. We establish now the analogous set of results in a continuous time setting. We start with an example where a duality gap arises. 
\begin{example}
\label{ex:dualitygap}

Consider $d=1$ and $\P_0\in \Qset$ be such that $X_t\equiv 1$, $t\in [0,1]$, i.e., $\P_0$ is a constant stock price model. Consider also $\P_1$ under which $X_t\equiv 1$, for $t\in [0,1/2]$ and then follows the SDE
\[
dX_t = \bar\beta X_t dW_t, 
\]
until $X$ reaches $100$ or $t=1$, whichever comes first, and where $\bar\beta>0$ is given and $W_t$ is a Brownian motion (possibly from a larger space). 
In other words, the stock price is constant for half of the time and then behaves as a scaled geometric Brownian motion which stops if it hits the level $100$.

Let the set $\Qset^E$ be the convex hull of $\P_0,\P_1$, so 
\[
E_1=\{X_{[0,1/2]}=1\},\quad E_2=\{\beta_{[1/2,1]}=0\}\cup\{\beta_t=\bar\beta^2\omega_t^2\mathbf{1}_{\sup_{s\leq t}\omega_s<100}, t\in[1/2,1]\},
\]
where $\beta$ is the diffusion characteristic of $X$.
In particular, $\tilde\P\in \Qset^E$, where $\tilde\P = \frac{1}{2} \P_0 + \frac{1}{2} \P_1$. Finally, we let $g = (X_1-1)^+-p$, where $p$ is such that $\E^{\tilde\P}[(X_1-1)^+] = p = \frac{1}{2}\E^{\P_1}[(X_1-1)^+]$, so that $\tilde\P\in \Qset^E_g$. 

We now define the American option payoff via $Z_t= 3p/2-6p|t-1/4|+(X_t - 1)^+\land 100$ for $t\in [0,1/2] $ and $Z_t= (X_t - 1)^+\land 100$ for $t\in (1/2,1]$. Clearly, by definition, $Z\in C_b(\Lambda)$. Note also that for any $\P\in \Qset^E$, we have $X_1\leq 100$ $\P$-a.s., so that $Z_1 = (X_1-1)^+$ $\Qset^E$-q.s., and $Z_{1/4}=3p/2$, $\Qset^E$-q.s. In fact, for any $\P\in \Qset^E$, $Z_t$ is a deterministic process on $[0,1/2]$ which attains its maximum at time $t=1/4$. In addition, for $t\in [1/2,1]$, $Z_t=f(X_t)$, for a convex function $f(w)=(w-1)^+$ and since $\P$ is a martingale measure, $Z_t$ is a submartingale on $[1/2,1]$ and, in particular, $\E^\P[Z_\tau] \leq \E^\P[Z_1]$, for any stopping time $\tau$ taking values in $[1/2,1]$. It follows that the American option will only ever be optimally exercised at times $t=1/4$ or $t=1$. In particular, $Z_{1/4}>Z_t$ for all $t\neq 1/4$ $\P_0$-a.s., and under this measure the optimal exercise time is $\tau_0=1/4$. In contrast, under $\P_1$, we have $\E^{\P_1}[Z_1\mid \Filt_{1/4}] = 2p>Z_{1/4}$ and the optimal exercise policy is $\tau_1\equiv 1$. However, under any $\P\in \Qset^E_g$ we have $\E^{\P}[Z_1\mid \Filt_{1/4}^\P] = \E^{\P}[Z_1]= p < Z_{1/4}$ and hence $\tau_0$ is the optimal exercise policy:
$$ \sup_{\P\in \Qset^E_g} \sup_{\tau\in \Tset^\P} \E^\P Z_\tau = 3p/2.
$$

We will now show that this value is strictly smaller than the superhedging price. In fact, we will show that the above supremum can be strictly increased by allowing a richer family of stopping times. Specifically, consider a two-dimensional standard Brownian motion $(B,W)$ on some probability space $(S,\mathcal{S},\mu)$. 
Let $\beta_t = 0$ for $t\in [0,1/2]$ and 
$$\beta_t =
\begin{cases}
    0 & \text{if } B_{1/4}<0,\\
    \bar\beta^2 X_t^2 \mathbf{1}_{\sup_{s\leq t}X_s<100} & \text{if } B_{1/4}\geq 0, 
\end{cases} \text{ for }t\in (1/2,1],
$$
where $X_0=1$ and $dX_t = \sqrt{\beta_t} dW_t$. Observe that the distribution of $X$ under $\mu$ is given by $\tilde \P$, 
$\tilde \P= \mu\circ X^{-1}$. Consider now $Y_t = \E^\mu[g\mid \cG_t]$, where $\cG_t = \sigma(B_s,W_s: s\leq t)^\mu$. Note that $Y$ is a continuous martingale and that $Y_0 = 0$ while $Y_{1/4} = -p\mathbf{1}_{B_{1/4}<0} + p\mathbf{1}_{B_{1/4}\geq 0}$. Let $\hat\P:= \mu\circ (X,Y)^{-1}$ which is a martingale measure on the canonical space with $d=2$. Note that $\tau_* = \frac{1}{4} \mathbf{1}_{Y_{1/4}<0} + \mathbf{1}_{Y_{1/4}>0}$ is a stopping time in its natural filtration and we have 
$$ \E^{\hat \P}[Z_{\tau_*}] = \frac{1}{2} 3p/2 + \frac{1}{2} 2p = 7p/4 > 3p/2.
$$
It will follow directly from Proposition \ref{prop:middleinequalities} below that this provides a lower bound on $\pi^A_{g,E}(Z)$ and hence there is a strict inequality in the weak duality in Proposition \ref{prop:Americanweakduality}.

\end{example}
To address the issue identified in the above example, in analogy to the discrete time results in \cite{aksamit2019robust}, we consider dynamic \emph{lifts} in which options $g$ are traded for all time $t\in [0,1]$. This is a technical device, a fictitious lift as first pioneered by \cite{CvitanicKaratzas} in the context of trading constraints, which is needed to establish a suitable pricing-hedging duality. The intuition is that these option prices can provide an additional signal which allows to build optimal exercise policies. However, an additional technical difficulty arises in continuous time in relation to how and when additional information is revealed. In Example \ref{ex:dualitygap} above we used an additional Brownian motion to, essentially, generate a flip of a coin at time $t=1/2$. This allowed us to ensure the option prices $Y$ were continuous and $Y_0=0$ was a constant. However, this is not always the case. For a trivial counterexample, suppose $\P$ is a convex combination of two Black-Scholes models with volatilities $\sigma_1<\sigma_2$. This can be realised by flipping an independent coin at time $t=0$ and deciding which model is used. In such a setup, option price $Y_0$ defined via conditional expectations (and taken right-continuous) would be random at time $t=0$ (taking two values). Naturally, if $\P$ is calibrated then we have $\E^\P[Y_0]=0$ and we simply have to mandate that trading in options happens at time ``$t=0-$", i.e., before the information is revealed, we refer to \cite{AksamitOblojZhou} for a detailed discussion. We implement this by working on a time horizon $[-\delta, 1]$ with all processes required to be constant on $[-\delta,0)$, where $\delta\in (0,1)$ is small. 

We thus introduce the enlarged space $\widehat\Omega=D([-\delta,1];\R^{d+m})$ with the canonical process $\widehat X=(X,Y)$ and their natural filtration $\widehat\FF$. We write $\FF^X$ for the filtration on $\widehat \Omega$ generated by $X$. For any probability measure $\widehat\P$ on $\widehat\Omega$, let $\widehat\FF^{\widehat\P}$ be the augmented filtration.   
Also let $\widehat\Tset$ (resp. $\widehat\Tset^{\widehat\P}$) be the set of $\widehat\FF$-stopping times (resp. $\widehat\FF^{\widehat\P}$-stopping times). The space of stopped paths on $\widehat\Omega$ is denoted $\widehat\Lambda$. Spaces $\bar{\widehat\Omega}$ and $\bar{\widehat\Lambda}$ are defined as previously, with $\widehat\Omega$ (resp.\ $\widehat\Lambda$) replacing $\Omega$ (resp.\ $\Lambda$). All these spaces are now defined for time index $t\in [-\delta, 1]$. 

We note that all of the results in section \ref{sec:robustpricing} can be applied in the context of martingale measures on $\widehat\Omega$. Further, we will now make use of the additional pathwise restrictions introduced there to lift calibrated martingale measures on $\Omega$ to restricted martingale measures on $\widehat\Omega$. Given $E_1\in \Borel(\Omega)$ and $E_2\in \Borel(\Lambda)$ we let 
\begin{align*}
    \widehat E_1 &:= \left\{\widehat{\omega}\in \widehat\Omega: \omega_{|[0,1]}\in E_1,\ X_t(\widehat\omega)=X_0(\widehat\omega), Y_t(\widehat\omega)=0, t\in [-\delta, 0),\ Y_1(\widehat\omega) = g(X(\widehat\omega))\right\}\in \Borel (\widehat \Omega),\\
    \widehat E_2 &:= \left\{(t,X_{\cdot \land t}, Y_{\cdot \land t})\in \widehat \Lambda: (t,X_{\cdot \land t})_{t\in [0,1]}\in E_2\right\}\in \Borel(\widehat \Lambda).
\end{align*}
Note that $X_t(\widehat\omega)=X_0(\widehat\omega), Y_t(\widehat\omega)=0, t\in [-\delta, 0)$ is a measurable constraint as it is enough to ensure this for rational $t$ thanks to right-continuity of paths. We write $\widehat{\Qset}^{\widehat E}=\{\widehat \P\in\widehat\Qset: \widehat\P(\widehat E_1)= \lambda\otimes\widehat\P (\widehat E_2)=1\}$. With this notation, we show that  every model in $\Qset^E_g$ can be lifted to a model in $\widehat\Qset^{\widehat E}$.
\begin{lemma}\label{lemmaliftXY}
For any $\widehat\P\in\widehat\Qset^{\widehat E}$, $\widehat \P \circ \left((X_t)_{t\in [0,1]}\right)^{-1} \in \Qset^E_g$. Conversely, for any $\P\in\Qset^E_g$, there exists a  $\widehat\P\in\widehat\Qset^{\widehat E}$ such that the $\widehat \P \circ \left((X_t)_{t\in [0,1]}\right)^{-1}=\P$.
\end{lemma}
\begin{proof}
    Let $\widehat\P\in\widehat\Qset^{\widehat E}$ and define $\P:=\widehat\P \circ \left((X_t)_{t\in [0,1]}\right)^{-1}$. Since a martingale remains a martingale in its natural filtration, we have $\P\in \Qset$. In addition, by definition of sets $\widehat E_1$, $\widehat E_2$, $\P\in \Qset^E$. Finally, since $Y$ is a $\widehat\P$-martingale with $Y_{-\delta} = 0$, we have $\E^\P[g(X(\omega))] = \E^{\widehat \P}[g(X(\widehat \omega))] =\E^{\widehat \P}[Y_1(\widehat \omega)]=0$ which establishes $\P\in \Qset^E_g$, as required. 
    To show the converse, fix $\P\in\Qset^E_g$. Define $g^0_t := \E^\P[g\mid \Filt^\P_t]$ and let $(g_t)_{t\in [0,1]}$ be the rcll martingale modification of $(g^0_t)_{t\in [0,1]}$. Note that $g_1=g(X)$ $\P$-a.s.\ and $\E^\P[g_0]=0$. Extend processes $X_\cdot$ and $g_\cdot$ to $[-\delta,1]$ via $X_t\equiv X_0$ and $g_t=0$ for $t\in [-\delta, 0)$, noting they remain martingales on $[-\delta, 1]$. We now define $\widehat \P := \P \circ (X_\cdot, g_\cdot)^{-1}$. The martingale property gives us $\widehat P\in \widehat \Qset$, by definition $\widehat\P(\widehat E_1)= \lambda\otimes\widehat\P (\widehat E_2)=1$ and $\widehat \P \circ \left((X_t)_{t\in [0,1]}\right)^{-1}=\P$, so that $\widehat \P\in \widehat\Qset^{\widehat E}$, as required.    
\end{proof}

Next, let us extend the notion of superhedging price to the extended model. This requires us to lift the American option payoff $Z$ to $\widehat \Omega$. On $[0,1]$ this is done via the obvious lift, $\widehat Z((t,\widehat\omega_{\cdot \land t})) := Z((t,X(\widehat \omega)_{|[0,t]}))$, but we also need to define the American payoff on $[-\delta, 0)$ and we do it in such a way that if $\widehat X$ is constant on $[-\delta,0)$ $\P$-a.s., then so is the American option payoff: $\widehat Z(t,\widehat \omega_{\cdot \land t}) := Z((0,X(\widehat \omega)_{t})), t\in [-\delta, 0)$. The superhedging price for $\widehat Z$ is defined in analogy to $\pi^A_{g,E}(Z)$. 
  
\begin{align*}
\widehat\pi^A_{\widehat E}(\widehat Z):=\inf\bigg\{ & x: \exists (q,\tilde q)\in C_b(\widehat\Lambda;\R^{d+m})\times C_b(\bar{\widehat\Lambda}_{\geq};\R^{d+m}),\ \text{s.t.}\  x+\int_{-\delta}^u q(t,\widehat\omega_{t\land \cdot})d\widehat X_t(\omega)\\
& \quad +\int_u^1 \tilde q^u(t,\widehat \omega_{t\land \cdot}) d\widehat X_t(\omega) \geq \widehat Z(u,\widehat \omega_{\cdot\wedge u}),\ \forall\, u\in[-\delta,1],\  \widehat\Qset^{\widehat E}(d\widehat\omega)\text{-\qs}, \bigg\}
\end{align*}

\begin{proposition}\label{prop:middleinequalities}
For $E_1\in \Borel(\Omega)$, $E_2\in \Borel(\Lambda)$, $Z\in C_b(\Lambda)$, $g\in C_b(\Omega; \R^m)$ such that $\Qset^E_g\neq \emptyset$, we have the following chain of inequalities:
\begin{gather}
\pi^A_{g,E}(Z) \geq \widehat\pi^A_{\widehat E} (\widehat Z)  \geq\sup_{\widehat\P\in \widehat\Qset^E, \widehat\tau\in\widehat\Tset^{\widehat\P}} \E^{\widehat\P} \widehat Z_{\widehat\tau} = \sup_{\bar{\widehat\P}\in\bar{\widehat\Qset}^E} \E^{\bar{\widehat\P}} \widehat Z (\theta, \widehat\omega_{\cdot \land \theta})
  \geq \sup_{\bar\P\in\bar\Qset^E_g} \E^{\bar\P} \bar Z.
\end{gather}
\end{proposition}

\begin{proof}
First, focus on the inequality  $\pi^A_{g,E}(Z) \geq \widehat\pi^A_{\widehat E} (\widehat Z)$. Consider any superhedging strategy $(x,q,\tilde q,h)$ for $\pi^A_{g,E}$. We need to show that we can lift it to a superhedging strategy on $\widehat \Omega$ for $\widehat Z$. We do this by simply buying and holding $X$ on $[-\delta, 0]$, in accordance to the initial position of $q(0,X_{-\delta}(\widehat \omega))$ and then trading in $X$ according to $(q,\tilde q)$ on $[0,1]$ as functions of $(t,X(\widehat \omega)_{|[0,t]})$. 
If the option holder exercises on $[\delta,0)$ then we do not change the strategy and respond with $\tilde q^0$ from time $t=0$ onwards. 
We do no dynamic trading in $Y$ but only buy and hold $h$ of $Y$ on $[-\delta, 1]$. 
The superhedging property is easily obtained with a contradiction argument, noting that  for any $\widehat \P\in \widehat \Qset^{\widehat E}$, $X$ and $Z$ are constant on $[-\delta, 0]$ and $Y_{-\delta}=0$, $\widehat \P$-a.s., and using the first part of Lemma \ref{lemmaliftXY}. 


Next, $\widehat\pi^A_E (\widehat Z)  \geq\sup_{\widehat\P\in \widehat\Qset^E, \widehat\tau\in\widehat\Tset^{\widehat\P}} \E^{\widehat\P} \widehat Z_{\widehat\tau}$ is the weak duality analogous to Proposition \ref{prop:Americanweakduality}, but applied to the extended model. It is immediate to check that the proof extends to this setting. Note that $\Qset^E_g$ is non-empty and hence, by Lemma \ref{lemmaliftXY}, so is $\widehat\Qset^E$. Analogously, the equality $\sup_{\widehat\P\in \widehat\Qset^E, \widehat\tau\in\widehat\Tset^{\widehat\P}} \E^{\widehat\P} \widehat Z_{\widehat\tau} = \sup_{\bar{\widehat\P}\in\bar{\widehat\Qset}^E} \E^{\bar{\widehat\P}} \widehat Z (\theta, \widehat\omega_{\cdot \land \theta})$ follows from Theorem \ref{thm:americanconvexmtg}, also applied to the extended model.

Finally, recall that Lemma \ref{lemmaliftXY} states that every measure in $\Qset^E_g$ can be lifted to a measure in $\widehat\Qset^E$. Similarly, and with obvious definitions, the same arguments show that every measure in $\bar\Qset^E_g$ can be lifted to a measure in $\bar{\widehat\Qset}^E$. Since $\widehat Z$ only depends on $(t,\omega_{\cdot\wedge t})$, the last inequality follows:
$\sup_{\bar{\widehat\P}\in\bar{\widehat\Qset}^E} \E^{\bar{\widehat\P}} \widehat Z (\theta, \widehat\omega_{\cdot \land \theta})
  \geq \sup_{\bar\P\in\bar\Qset^E_g} \E^{\bar\P} Z (\theta, \omega_{\cdot \land \theta})$.
\end{proof}
\begin{remark}
    We note that the assumptions we made on $Z$ and $g$ were stronger than required for the results in this section. We made no use of continuity of $g$ and we used that it is bounded to conclude that its conditional expectations defined a square integrable martingale. In fact, all results in section \ref{sec:robusthedging} extend instantly to $g\in L^0(\Omega)$ with $\Qset^E_g$ redefined as 
    $$\left\{\P\in \Qset^E: \E^\P(g^2)<\infty\text{ and }\E^\P(g)=0\right\}$$
and assumed non-empty. Similarly, we did not use that $Z\in C_b(\Lambda)$ and it would have been enough to assume that $Z\in L^0(\Lambda)$, bounded from below and such that $t\to Z(t,\omega_{\cdot \land t})$ is rcll $\Qset^E(d\omega)$-q.s. We made stronger assumptions in the statements as this, we feel, streamlines the presentation allowing to match the assumptions to those made in the next section, where they are actually needed for the results, and to obtain the robust pricing-hedging duality. 
\end{remark}

\section{Robust pricing-hedging duality for American options}
\label{sec:robustPHduality}
The previous sections provided the connection between the robust prices of American options and the robust prices of European options in the enlarged space. In order to obtain duality for American options, we must first establish the robust pricing-hedging duality for European options in the enlarged space.

To this end, we restrict our probability space to continuous paths, i.e., $\Omega=C([0,1],\R^d)$. All associated objects, such as $\FF, \Qset, \bar\Omega, \bar\FF, \bar\Qset$ and so on, are defined in the same way as before.
In terms of the family of candidate models, we require some restrictions on the quadratic variation of the underlying asset.
Similarly to \cite{SonerTouziZhang}, define 
\begin{equation}\label{eq:defQV}
     \langle X\rangle_t := X_t X_t^\intercal - 2\int_0^t X_u dX_u, \quad \text{and}\quad \beta_t:= \limsup_{n\to \infty}\frac{\langle X\rangle_{t}- \langle X\rangle_{t-2^{-n}}}{2^{-n}}, \quad t\in (0,1],
\end{equation} 
with $\langle X\rangle_{0}=\beta_0=0$. We let $\Omega_{QV,t}\subset \Omega_t$ be the set of paths on which $(X\bcdot X)_{s\in [0,t]}$ in \eqref{eq:defSI} is the uniform limit of its approximations. 
By the above, $\Omega_{QV,t}\in \Filt_t$. Note that the set of absolutely continuous functions is a Borel subset of all continuous functions and hence the further subset $\Omega_{QV,t}^{AC}\subset \Omega_{QV,t}$ of paths on which $\langle X\rangle_{s\in [0,t]}$ is an absolutely continuous function is measurable, $\Omega_{QV,t}^{AC}\in \Filt_t$, $t\in (0,1]$. Similarly, $\beta_{s\in [0,t]}$ is a limsup of measurable functions (defined on $(2^{-n},1]$ and extended by continuity to $[0,1]$) and hence is measurable, so that the process $\beta$ is progressively measurable and we have $\langle X\rangle_1 = \int_0^1 \beta_t dt$ on $\Omega_{QV,1}^{AC}$.

If $X$ is a local martingale under some $\P\in \Pset$ then, by \cite[Thm.~2]{Karandikar}, we see that $\langle X\rangle$ is its quadratic variation (such a process being normally only defined $\P$-a.s.) and that for any $q\in C_b(\Lambda)$ the pathwise stochastic integral in \eqref{eq:defSI} coincides with the usual It\^o integral $\P$-a.s. In addition, since $\beta$ is progressively measurable and $\Omega_{QV,1}^{AC}\in\Filt_1$, we can define the set $\Qset$ of measures $\P\in \Pset$ such that $X$ is a square integrable $\P$-martingale and $\langle X\rangle_{\cdot}$ is absolutely continuous with $\beta_t\in \S^d_+$ for all $t\in [0,1]$, $\P$-a.s. Note that $X$ is a square integrable $\P$-martingale if and only if $X$ is a $\P$-local martingale with  $\E^\P\left(\int_0^1 |\beta_t|\, dt\right)< \infty.$ 

\newcommand{\volE}{\mathfrak{E}}
Often it is useful to restrict the diffusion characteristic $\beta$ of $X$ to take value in a subset that depends on $(t,\omega_{\cdot\wedge t})$. We impose the following market assumptions.
\begin{assumption}\label{assumptionD}
The initial asset prices are given by $x_0\in \R^d$. Volatility is restricted by a set-valued process $\volE:\Lambda\to\mathcal{B}(\S^d_+)$ with $\volE(t,\omega_{\cdot\wedge t})\subseteq \S^d_+$ closed, convex, and globally bounded for $(t,\omega_{\cdot\wedge t})\in \Lambda$ and such that 
$\volE$ is continuous with respect to the topology induced by the Hausdorff metric on compact subsets.
We let $g\in C_b(\Omega; \R^m)$,
$$E_1 = \{\omega\in\Omega: \omega_0=x_0\}, \quad E_2 = \{(t,\omega_{t\land \cdot}): \beta(t,\omega_{t\land \cdot})\in \volE(t,\omega_{t\land \cdot})\}
$$
and define the pricing measures via \eqref{eq:QsetEdef}. We assume that $\Qset^E_g\neq \emptyset$.
\end{assumption}

We now present the main robust pricing-hedging duality result for American options. We recall the notation associated to the dynamic lift, as introduced in section \ref{sectiondualitygap}. 
\begin{theorem}\label{thm:americanduality}
 Suppose Assumption \ref{assumptionD} holds and $Z\in C_b(\Lambda)$. Then we have equalities throughout in Proposition \ref{prop:middleinequalities} and in particular
\[
\pi^A_{g,E}(Z) = \sup_{\widehat\P\in \widehat\Qset^{\widehat E}, \widehat\tau\in\widehat\Tset^{\widehat\P}} \E^{\widehat\P} {\widehat Z}_{\widehat\tau}.
\] 
\end{theorem}
\begin{remark}
In the case that there are no static European hedging instruments, Theorem \ref{thm:americanduality} reduces to a robust pricing-hedging duality of the form
\[
\pi^A_E(Z) = \sup_{\P\in \Qset^E, \tau\in\Tset^{\P}} \E^{\P} Z_\tau.
\]
While our result work for general restrictions $\volE$ on volatility, and this includes the special case of $\volE$ being a singleton. In particular, we can have $\Qset^E=\{\P\}$ be a singleton and recover the classical pricing-hedging duality of \cite{Myneni}, for the case of a bounded, continuous volatility model.
\end{remark}
The proof of this theorem identifies the superhedging price of an American option on $\Omega$ with that of a European option on an enlarged space $\bar\Omega$ and establishes pricing-hedging duality for the latter. This, in turn, is possible since the restrictions on the diffusion characteristic in Assumption \ref{assumptionD} allow for the application of the semimartingale optimal transport duality from \cite{guo2021path} to establish European option dualities. In particular, if a cost function returns $0$ when  $\beta_t\in \volE(t,\omega_{t\land \cdot})$ and returns $+\infty$ otherwise, then it would be convex, lower semicontinuous and coercive, while its convex conjugate would be continuous, which are the required properties.

\subsection{Duality for European options}

The financial interpretation of the setup introduced so far is standard. The process $X$ models price dynamics of $d$ assets traded continuously and in which investors are allowed to trade in a frictionless manner. We assume interest rates are deterministic and all prices are given in discounted units. Likewise, European or other payoffs, are expressed as functions of the discounted asset prices. 

We start with a brief summary of the relevant robust pricing and hedging results for European options. As recalled in the introduction, many variants of continuous time pricing-hedging duality exist and they differ by assumptions they impose, the hedging strategies they consider and the regularity of the option payoff they can cover. Here, we adopt the OT-driven framework of \cite{guo2021path}. We consider a function $g\in C_b(\Omega; \R^m)$ which represents a vector of liquidly traded payoffs shifted by their market prices, so that a pricing measure $\P$ is calibrated if $\E^\P(g) = 0$. We denote  $\Qset^E_g\subset\Qset^E$ the subset of calibrated pricing measures.

Suppose $\Qset^E_g$ represents all the pricing models we consider. Then for a European payoff $f$ its \emph{robust model price} is given by
\[
\sup_{\P\in\Qset^E_g} \E^\P f.
\]
On the other hand, we can consider a pricing-via-hedging argument, and look for the \emph{superhedging price} defined as the smallest initial capital such that it is possible to construct a portfolio which superreplicates the payoff $f$. Portfolios to consider can trade dynamically in the stock and also statically (buy-and-hold at time $t=0$) in the vector $g$ of options with fixed market prices. In addition, the superreplication property needs to hold $\Qset^E_g$ quasi-surely (q.s.), that is $\P$-a.s., for all $\P\in \Qset^E_g$. So the superhedging price is given by
\begin{equation}\label{eq:EuroSH}
    \pi_{g,E}(f):=\inf\left\{x: \exists (q,h)\in C_b(\Lambda;\R^d)\times \R^m ,\  \text{s.t.}\  x+\int_0^1 q_t dX_t + h g \geq f,\ \Qset^E_g\text{-\qs}\right\}.
\end{equation}
The key result, known as the \emph{robust pricing-hedging duality}, asserts that these two approaches to computing the robust price for $f$ are consistent and give the same result. 
\begin{proposition}[\cite{guo2021path}] 
\label{prop:Euroduality}
Suppose Assumption \ref{assumptionD} holds. Then for any $f\in C_b(\Omega)$, 
$\pi_{g,E}(f) = \sup_{\P\in\Qset^E_g} \E^\P [f]$.
\end{proposition}
Note that we have restricted the hedging strategies to $q\in C_b(\Lambda;\R^d)$ and, in particular, both sides of the superhedging inequality in the definition of $\pi_{g,E}(f)$ make sense pathwise on $\Omega$ thanks to \eqref{eq:defSI}. Restricting to continuous integrands is sufficient since we only consider payoffs $f\in C_b(\Omega)$. For more general $f$'s, we would need a larger class of hedging strategies, e.g., \cite{DenisMartini,OblojZhou}. This has been researched in detail in the past, see the discussion in section \ref{sec:intro} above. On the other hand, using only $q\in C_b(\Lambda;\R^d)$ means we do not have to worry about fine questions of admissibility: since $q$ is bounded, the stochastic integral $\int_0^t q_s dX_s$, $t\in [0,1]$, is a square integrable martingale under any $\P\in \Qset$, and in particular its expectation is zero. 



As shown in the previous section, the robust price of an American option is equal to the robust price of the corresponding European option in the enlarged space.
Hence, in order to establish dualities for American options, we will first establish dualities for European options in the enlarged space, in an analogue of Proposition \ref{prop:Euroduality}.

\subsection{Duality for European options in the enlarged space}

Under any probability measure $\P$ on $\bar \Omega$, $\vartheta$ is a semimartingale with characteristics $(\I_{[0,\theta]}(t), 0)$ and the integral $t \to \int_0^t q_s d\vartheta_s$ is defined pathwise on $\bar \Omega$ for $q\in C(\bar\Lambda)$. 
 We let $\bar \Qset$ be the set of probability measures $\P$ on $\bar\Omega$ such that $\P$ restricted to $(\Omega, \FF_1)$ is in $\Qset$ and such that $X$ is a $(\P,\bar \FF)$-martingale. We recall that $\langle X\rangle$ and $\beta$ were defined pathwise in \eqref{eq:defQV} and these definitions extend to $\bar \Omega$ with 
$\langle X\rangle(\bar \omega) = \langle X\rangle (\omega)$ and hence $\beta(\bar \omega) = \beta(\omega)$. 
For restrictions on $\beta$, we consider Assumption \ref{assumptionD}. Note that $\vartheta$ automatically has a diffusion characteristic of 0, and we also don't allow $\volE$ to depend on $\vartheta$. Then the set $\bar \Qset^E\subset \bar Q$, of martingale measures where $\beta\in \volE$, and its subset $\bar \Qset^E_g\subset \bar \Qset^E$ of calibrated measures, are clearly defined. 

The definition of the \emph{superhedging price} in \eqref{eq:EuroSH} naturally extends to $\bar \Omega$:
\[
\bar\pi_{g,E}(f):=\inf\{x: \exists (q,h)\in C_b(\bar\Lambda;\R^d) \times \R^m,\  \text{s.t.}\  x+\int_0^1 q_t dX_t + h g \geq f,\ \bar\Qset^E_g\text{-\qs}\}.
\]
When the vector $g$ is empty we will simply write $\bar\pi_E(f)$.

\begin{proposition}\label{enlargedduality}
Suppose Assumption \ref{assumptionD} holds and $f\in C_b(\bar\Omega)$.
Then
$\bar\pi_{g,E}(f) = \sup_{\bar\P\in\bar\Qset^E_g} \E^{\bar\P} f$.
\end{proposition}
\begin{proof}
By taking expectations, it is easy to check that
\[
\bar\pi_{g,E}(f) \geq \sup_{\bar\P\in\bar\Qset^E_g} \E^{\bar\P} f.
\]
To show the reverse inequality we apply the main duality result of \cite{guo2021path} on the enlarged spaces $\bar\Omega$ and $\bar\Lambda$, with a cost function that 
equals 0 if $\beta(t,\omega_{\cdot\wedge t})\in \volE(t,\omega_{\cdot\wedge t})$, or equals infinity otherwise.
The duality gives us that
\begin{gather}
\sup_{\bar\P\in\bar\Qset^E_g} \E^{\bar\P} f = \inf_{h\in\R^m,\phi\in C^{1,1,2}_0(\bar\Lambda)}  \phi(0,0,X_0),\label{eqthmduality}\\
\text{subject to}\quad \phi(1,\cdot,\cdot)\geq f-hg \quad \text{and}  \quad \Dt\phi + \I(t\leq\theta) \Dth \phi + \sup_{\beta\in \volE} \frac{1}{2}\beta: \Dxx\phi\leq 0, \label{eqthmduality032}
\end{gather}
where the set of test functions is defined as follows.  
We say $\phi\in C^{1,1,2}(\bar\Lambda)$ if $\phi\in C_b(\bar\Lambda)$ and there exist functions $(\Dt \phi, \Dth \phi,\Dx \phi, \Dxx \phi)\in C_b(\bar\Lambda;\R^{d+2}\times \S^d)$ such that, for any $\bar\P\in\bar\Qset$ and $u\in[0,1]$, the following \emph{functional It\^o formula} holds:
\begin{align}\label{eqfunctionalito01}
\phi(u,\vartheta,X)-\phi(0,\vartheta,X)&=\int_0^u \Dt \phi\, ds+ \Dth \phi\,  d\vartheta_s+ \Dx \phi \, dX_s + \frac{1}{2} \Dxx \phi : d\langle X\rangle_s\\
&=\int_0^u (\Dt \phi+  \I_{[0,\theta]} \Dth \phi)\, ds+ \Dx \phi \, dX_s + \frac{1}{2} \Dxx \phi : d\langle X\rangle_s, \quad \bar\P\text{-\as} \nonumber
\end{align}
From this definition, it follows directly that for each $\phi$ satisfying \eqref{eqthmduality032} and each $\bar\P\in \bar\Qset^E$, the following holds $\bar\P$-\as
\begin{align*}
f-hg -\phi(0,0,X_0) \leq & \ \phi(1,\cdot,\cdot)-\phi(0,0,X_0)\\
&=\int_0^1 (\Dt \phi+ \I_{[0,\theta]} \Dth \phi+ \,\frac{1}{2} \beta^\P:\Dxx \phi) dt+ \Dx \phi \, dX_t\\
&\leq \int_0^1 \Dx \phi \, dX_t.
\end{align*}
As $\Dx \phi\in C_b(\bar\Lambda;\R^d)$ is an admissible hedging strategy, we deduce that $\phi(0,0,X_0) \geq \bar\pi_{g,E}(f)$. Since this holds for all $\phi$ satisfying \eqref{eqthmduality032}, it implies $ \sup_{\bar\P\in\bar\Qset^E_g} \E^{\bar\P} f \geq \bar\pi_{g,E}(f)$, completing the proof.
\end{proof}

\begin{remark}
   We defined the set $C^{1,1,2}(\bar\Lambda)$ by requiring that the functional It\^o formula holds for all $\bar\P\in \bar\Qset$, which is a smaller set than the set of all  semimartingale measures with integrable characteristics used in \cite{guo2021path}. This means that the equality in \eqref{eqthmduality} shown therein could, \emph{a priori}, turn to an inequality ``LHS$\geq$ RHS" with our definition. However, the rest of the proof then shows that the equality in fact holds through the sandwiching with the first inequality.  
\end{remark}

\begin{remark}\label{rk:hedgingstrategies}
    It follows from the proof that in the definition of the superhedging price, we can restrict further the hedging strategies and instead of $(q,h)\in C_b(\bar\Lambda;\R^d) \times \R^m$ it suffices to consider $(q,h)\in C_x(\bar\Lambda;\R^d) \times \R^m$, where
       $$C_x(\bar\Lambda;\R^d)= \left\{ \Dx \phi : \phi\in C^{1,1,2}(\bar\Lambda)\right\}.
       $$
    We also note that for such $(q,h)\in C_x(\bar\Lambda;\R^d) \times \R^m$, the functional It\^o formula can be used to \emph{define} the integral $\int_0^{\cdot} q_s dX_s$ pathwise on $\bar\Omega$. This definition would potentially differ from \eqref{eq:defSI} but the two agree $\bar\P$-a.s., for any $\bar\P\in \bar\Qset$. 

 \end{remark}


\subsection{Duality for American options}

Finally, we can return to the problem of superhedging of American options and link it to the superhedging for European options in the enlarged space. Recall that for $Z\in C_b(\Lambda)$ we write $\bar Z(\bar \omega) = Z(\theta,\omega_{\cdot \land \theta})$. 
Following Remark \ref{rk:hedgingstrategies} we observe that in the setup of Proposition \ref{enlargedduality} we have 
\begin{align*}
\bar\pi_{g,E}(\bar Z)&=\inf\{x: \exists (\bar q,\bar h)\in C_x(\bar\Lambda) \times \R^m,\  \text{s.t.}\  x+\int_0^1 \bar q_t dX_t + \bar h g \geq \bar Z,\ \bar\Qset^E_g\text{-\qs}\}.
\end{align*}

We then have the following inequality.
\begin{lemma}\label{lemeuroamer} Suppose Assumption \ref{assumptionD} holds and $Z\in C_b(\Lambda)$. Then $\bar\pi_{g,E}(\bar Z) \geq \pi^A_{g,E}(Z)$.  
\end{lemma}
\begin{proof}Suppose $(\bar q, \bar h)$ is a superhedge for $\bar\pi_{g,E}(\bar Z)$.
We construct the following superhedge $(q,\tilde q,h)$ for $\pi^A_{g,E}(Z)$, noting that $\tilde q^\theta(t,\cdot)$ is only relevant for $t\geq \theta$.
\begin{align*}
q(t,\omega_{\cdot\wedge t}) &= \bar q(t,t,\omega_{\cdot\wedge t}),\\
\tilde q^\theta(t,\omega_{\cdot\wedge t}) & = \bar q(t,\theta,\omega_{\cdot\wedge t}),\quad t\geq \theta,\\
h &= \bar h.
\end{align*}
It remains to check that $(q,\tilde q,h)$ is indeed a superhedge for $\pi^A_{g,E}(Z)$.
For each rational $u\in[0,1]$ and $\P\in\Qset^E_g$, we can define $\bar\P=\delta_u\times \P$. It is easy to check that $\bar\P \in \bar\Qset^E_g$.
Hence
\[
x+\int_0^u q dX_t+\int_u^1 \tilde q^u dX_t + h g \geq Z_u, \quad \P\text{-\as}.
\]
To extend from rational $u$ to all reals, first note that $Z_u$ is continuous in $u$. On the other side, by the definition of $q$ and $\tilde q^u$,
\begin{align*}
\int_0^u q dX_t+\int_u^1 \tilde q^u dX_t = \int_0^1 \bar q(t,u\wedge t,\cdot) dX_t.
\end{align*}
Recall $\bar q\in C_x(\bar\Lambda)$, so there exists $\phi\in C^{1,1,2}(\bar\Lambda)$ such that $\bar q=\nabla_x \phi$. Then by the functional It\^o formula in \eqref{eqfunctionalito01}, 
\begin{align*}
\int_0^1 \bar q(t,u\wedge t,\cdot) dX_t = \phi(1,u,\cdot)-\phi(0,u,\cdot) - \int_0^1 (\Dt \phi+  \I_{[0,u]} \Dth \phi)\, ds + \frac{1}{2} \Dxx \phi : d\langle X\rangle_s, \quad \P\text{-\as}
\end{align*}
Since $\phi, \Dt\phi,\Dth\phi, \Dxx\phi$ are all bounded and continuous, and $X$ has a bounded diffusion characteristic under $\P$, the right hand side, and hence also the left hand side, is continuous in $u$.
Thus
\[
x+\int_0^u q dX_t+\int_u^1 \tilde q^u dX_t + h g \geq Z_u,\quad \forall\,u\in[0,1], \quad\P\text{-\as}
\]
Since this holds for all $\P\in\Qset^E_g$, it must be a $\pi^A_{g,E}(Z)$ superhedge. Thus $\bar\pi_{g,E}(\bar Z) \geq \pi^A_{g,E}(Z)$.
\end{proof}

%
%

\begin{proof}[Proof of Theorem \ref{thm:americanduality}]
The required result follows from the following chain of inequalities
\begin{gather}\label{eqamericanduality1}
\bar\pi_{g}(\bar Z) \geq \pi^A_{g}(Z) \geq \widehat\pi^A (\widehat Z)  \geq\sup_{\widehat\P\in \widehat\Qset^E, \widehat\tau\in\widehat\Tset^{\widehat\P}} \E^{\widehat\P} \widehat Z_{\widehat\tau} = \sup_{\bar{\widehat\P}\in\bar{\widehat\Qset}^E} \E^{\bar{\widehat\P}} \bar{\widehat Z}
  \geq \sup_{\bar\P\in\bar\Qset^E_g} \E^{\bar\P} \bar Z = \bar\pi_g(\bar Z),
\end{gather}
which implies that there must be equality throughout. We will now justify each step.

The first inequality $\bar\pi_g(\bar Z) \geq \pi^A_g(Z)$ was just established in Lemma \ref{lemeuroamer}, while the last equality $\bar\pi_g(\bar Z) = \sup_{\bar\P\in\bar\Qset^E_g} \E^{\bar\P} \bar Z$ is the European duality in the enlarged space, proven in Proposition \ref{enlargedduality}. The remaining chain of inequalities was proven in Proposition \ref{prop:middleinequalities}.
\end{proof}


\bibliographystyle{plainnat}
\bibliography{bibfile}

\begin{thebibliography}{60}
\providecommand{\natexlab}[1]{#1}
\providecommand{\url}[1]{\texttt{#1}}
\expandafter\ifx\csname urlstyle\endcsname\relax
  \providecommand{\doi}[1]{doi: #1}\else
  \providecommand{\doi}{doi: \begingroup \urlstyle{rm}\Url}\fi

\bibitem[Aksamit et~al.(2019)Aksamit, Deng, Ob{\l}{\'o}j, and
  Tan]{aksamit2019robust}
Anna Aksamit, Shuoqing Deng, Jan Ob{\l}{\'o}j, and Xiaolu Tan.
\newblock The robust pricing--hedging duality for {A}merican options in
  discrete time financial markets.
\newblock \emph{Mathematical Finance}, 29\penalty0 (3):\penalty0 861--897,
  2019.

\bibitem[Aksamit et~al.(2020)Aksamit, Hou, and Obłój]{AksamitOblojZhou}
Anna Aksamit, Zhaoxu Hou, and Jan Obłój.
\newblock Robust framework for quantifying the value of information in pricing
  and hedging.
\newblock \emph{SIAM Journal on Financial Mathematics}, 11\penalty0
  (1):\penalty0 27--59, 2020.

\bibitem[Allan et~al.(2023)Allan, Liu, and Pr\"{o}mel]{Allan2023}
Andrew~L. Allan, Chong Liu, and David~J. Pr\"{o}mel.
\newblock A càdlàg rough path foundation for robust finance.
\newblock \emph{Finance and Stochastics}, 28\penalty0 (1):\penalty0 215--257,
  2023.

\bibitem[Az\'ema(1972)]{Azema1972Inventiones}
Jacques Az\'ema.
\newblock Quelques applications de la th\'eorie g\'en\'erale des processus.
  {I}.
\newblock \emph{Inventiones Mathematicae}, 18:\penalty0 293--336, 1972.

\bibitem[Baxter and Chacon(1977)]{BaxterChacon}
J.~R. Baxter and R.~V. Chacon.
\newblock Compactness of stopping times.
\newblock \emph{Zeitschrift f\"ur Wahrscheinlichkeitstheorie und Verwandte
  Gebiete}, 40\penalty0 (3):\penalty0 169--181, 1977.

\bibitem[Bayraktar and Zhou(2017)]{BayraktarZhou2017}
Erhan Bayraktar and Zhou Zhou.
\newblock Super-hedging {A}merican options with semi-static trading strategies
  under model uncertainty.
\newblock \emph{International Journal of Theoretical and Applied Finance},
  20\penalty0 (06):\penalty0 1750036, 2017.

\bibitem[Bayraktar et~al.(2015)Bayraktar, Huang, and Zhou]{BayraktarHuangZhou}
Erhan Bayraktar, Yu-Jui Huang, and Zhou Zhou.
\newblock On hedging {A}merican options under model uncertainty.
\newblock \emph{SIAM Journal on Financial Mathematics}, 6\penalty0
  (1):\penalty0 425--447, 2015.

\bibitem[Beiglb\"{o}ck et~al.(2013)Beiglb\"{o}ck, Henry-Labord\`{e}re, and
  Penkner]{BHLP}
Mathias Beiglb\"{o}ck, Pierre Henry-Labord\`{e}re, and Friedrich Penkner.
\newblock Model-independent bounds for option prices: a mass transport
  approach.
\newblock \emph{Finance and Stochastics}, 17\penalty0 (3):\penalty0 477--501,
  2013.

\bibitem[Beiglb\"{o}ck et~al.(2016)Beiglb\"{o}ck, Cox, and
  Huesmann]{BCHInventiones2017}
Mathias Beiglb\"{o}ck, Alexander M.~G. Cox, and Martin Huesmann.
\newblock Optimal transport and {S}korokhod embedding.
\newblock \emph{Inventiones mathematicae}, 208\penalty0 (2):\penalty0 327--400,
  2016.

\bibitem[Black and Scholes(1973)]{BlackScholes}
Fischer Black and Myron Scholes.
\newblock The pricing of options and corporate liabilities.
\newblock \emph{Journal of Political Economy}, 81\penalty0 (3):\penalty0
  637--654, 1973.

\bibitem[Bouchard and Nutz(2015)]{BN}
Bruno Bouchard and Marcel Nutz.
\newblock Arbitrage and duality in nondominated discrete-time models.
\newblock \emph{The Annals of Applied Probability}, 25\penalty0 (2):\penalty0
  823--859, 2015.

\bibitem[Breeden and Litzenberger(1978)]{BreedenLitzenberger}
Douglas~T. Breeden and Robert~H. Litzenberger.
\newblock Prices of state-contingent claims implicit in option prices.
\newblock \emph{Journal of Business}, pages 621--651, 1978.

\bibitem[Brown et~al.(2001)Brown, Hobson, and Rogers]{HobsonRogers}
Haydyn Brown, David Hobson, and L.~C.~G. Rogers.
\newblock Robust hedging of barrier options.
\newblock \emph{Mathematical Finance}, 11\penalty0 (3):\penalty0 285--314,
  2001.

\bibitem[Burzoni et~al.(2019)Burzoni, Frittelli, Hou, Maggis, and
  Obłój]{BFMOZ}
Matteo Burzoni, Marco Frittelli, Zhaoxu Hou, Marco Maggis, and Jan Obłój.
\newblock Pointwise arbitrage pricing theory in discrete time.
\newblock \emph{Mathematics of Operations Research}, 44\penalty0 (3):\penalty0
  1034--1057, 2019.

\bibitem[Cont and Fourni{\'e}(2013)]{ContFournier}
Rama Cont and David-Antoine Fourni{\'e}.
\newblock {Functional Itô calculus and stochastic integral representation of
  martingales}.
\newblock \emph{The Annals of Probability}, 41\penalty0 (1):\penalty0 109--133,
  2013.

\bibitem[Cox and Ob{\l}{\'o}j(2011)]{CoxObloj}
Alexander M.~G. Cox and Jan Ob{\l}{\'o}j.
\newblock Robust pricing and hedging of double no-touch options.
\newblock \emph{Finance and Stochastics}, 15\penalty0 (3):\penalty0 573--605,
  2011.

\bibitem[Cvitanic and Karatzas(1993)]{CvitanicKaratzas}
Jaksa Cvitanic and Ioannis Karatzas.
\newblock Hedging contingent claims with constrained portfolios.
\newblock \emph{The Annals of Applied Probability}, 3\penalty0 (3), 1993.

\bibitem[Davis et~al.(2018)Davis, Obł{\'o}j, and Siorpaes]{DavisOblojSiorpaes}
Mark Davis, Jan Obł{\'o}j, and Pietro Siorpaes.
\newblock {Pathwise stochastic calculus with local times}.
\newblock \emph{Annales de l'Institut Henri Poincaré, Probabilités et
  Statistiques}, 54\penalty0 (1):\penalty0 1--21, 2018.

\bibitem[Denis and Martini(2006)]{DenisMartini}
Laurent Denis and Claude Martini.
\newblock A theoretical framework for the pricing of contingent claims in the
  presence of model uncertainty.
\newblock \emph{The Annals of Applied Probability}, 16\penalty0 (2):\penalty0
  827--852, 2006.

\bibitem[Dolinsky(2014)]{Dolinsky2014GameOptions}
Yan Dolinsky.
\newblock {Hedging of game options under model uncertainty in discrete time}.
\newblock \emph{Electronic Communications in Probability}, 19:\penalty0 1--11,
  2014.

\bibitem[Dolinsky and Soner(2014)]{dolinsky2014martingale}
Yan Dolinsky and H~Mete Soner.
\newblock Martingale optimal transport and robust hedging in continuous time.
\newblock \emph{Probability Theory and Related Fields}, 160\penalty0
  (1-2):\penalty0 391--427, 2014.

\bibitem[Dunford and Schwartz(1958)]{dunford1958linear}
Nelson Dunford and Jacob~T. Schwartz.
\newblock \emph{Linear Operators {Part I}: {General Theory}}, volume~7 of
  \emph{Pure and Applied Mathematics}.
\newblock Interscience Publishers, New York, 1958.

\bibitem[Eckstein et~al.(2021)Eckstein, Guo, Lim, and
  Ob\l{}\'{o}j]{EcksteinGuoLimObloj}
Stephan Eckstein, Gaoyue Guo, Tongseok Lim, and Jan Ob\l{}\'{o}j.
\newblock Robust pricing and hedging of options on multiple assets and its
  numerics.
\newblock \emph{SIAM Journal on Financial Mathematics}, 12\penalty0
  (1):\penalty0 158--188, 2021.

\bibitem[El~Karoui and Tan(2013)]{ElKarouiTan}
Nicole El~Karoui and Xiaolu Tan.
\newblock Capacities, measurable selection and dynamic programming {Part II}:
  Application in stochastic control problems.
\newblock \emph{arXiv preprint arXiv:1310.3364}, 2013.

\bibitem[F{\"o}llmer(1981)]{Follmer}
Hans F{\"o}llmer.
\newblock Calcul d'{It{\^o}} sans probabilites.
\newblock In \emph{S{\'e}minaire de Probabilit{\'e}s XV 1979/80}, pages
  143--150, Berlin, Heidelberg, 1981. Springer Berlin Heidelberg.

\bibitem[F{\"o}llmer and Schied(2004)]{FollmerSchied2nd}
Hans F{\"o}llmer and Alexander Schied.
\newblock \emph{Stochastic Finance: An Introduction in Discrete Time}.
\newblock Walter de Gruyter, 2\textsuperscript{nd} edition, 2004.

\bibitem[Fremlin et~al.(1972)Fremlin, Garling, and Haydon]{fremlin1972bounded}
D.~H. Fremlin, D.~J.~H. Garling, and R.~G. Haydon.
\newblock Bounded measures on topological spaces.
\newblock \emph{Proceedings of the London Mathematical Society}, 3\penalty0
  (1):\penalty0 115--136, 1972.

\bibitem[Galichon et~al.(2014)Galichon, Henry-Labord{\`e}re, and Touzi]{GHLT}
Alfred Galichon, Pierre Henry-Labord{\`e}re, and Nizar Touzi.
\newblock A stochastic control approach to no-arbitrage bounds given marginals,
  with an application to lookback options.
\newblock \emph{The Annals of Applied Probability}, 24\penalty0 (1):\penalty0
  312--336, 2014.

\bibitem[Guo and Ob{\l}{\'o}j(2019)]{GuoObloj}
Gaoyue Guo and Jan Ob{\l}{\'o}j.
\newblock Computational methods for martingale optimal transport problems.
\newblock \emph{The Annals of Applied Probability}, 29\penalty0 (6):\penalty0
  pp. 3311--3347, 2019.

\bibitem[Guo et~al.(2017)Guo, Tan, and Touzi]{GuoTanTouzi}
Gaoyue Guo, Xiaolu Tan, and Nizar Touzi.
\newblock Tightness and duality of martingale transport on the {S}korokhod
  space.
\newblock \emph{Stochastic Processes and their Applications}, 127\penalty0
  (3):\penalty0 927--956, 2017.

\bibitem[Guo and Loeper(2021)]{guo2021path}
Ivan Guo and Gregoire Loeper.
\newblock Path dependent optimal transport and model calibration on exotic
  derivatives.
\newblock \emph{The Annals of Applied Probability}, 31\penalty0 (3):\penalty0
  1232--1263, 2021.

\bibitem[Hansen and Marinacci(2016)]{HansenMarinacci}
Lars~Peter Hansen and Massimo Marinacci.
\newblock Ambiguity aversion and model misspecification: An economic
  perspective.
\newblock \emph{Statistical Science}, 31\penalty0 (4):\penalty0 511--515, 2016.

\bibitem[Hobson(1998)]{Hobson}
David Hobson.
\newblock Robust hedging of the lookback option.
\newblock \emph{Finance and Stochastics}, 2\penalty0 (4):\penalty0 329--347,
  1998.

\bibitem[Hobson(2011)]{HobsonSurvey}
David Hobson.
\newblock The {S}korokhod embedding problem and model-independent bounds for
  option prices.
\newblock In \emph{Paris-{P}rinceton {L}ectures on {M}athematical {F}inance
  2010}, volume 2003 of \emph{Lecture Notes in Math.}, pages 267--318.
  Springer, Berlin, 2011.

\bibitem[Hobson and Neuberger(2017)]{HobsonNeuberger}
David Hobson and Anthony Neuberger.
\newblock Model uncertainty and the pricing of {A}merican options.
\newblock \emph{Finance and Stochastics}, 21\penalty0 (1):\penalty0 285--329,
  2017.

\bibitem[Hou and Ob\l\'oj(2018)]{OblojZhou}
Zhaoxu Hou and Jan Ob\l\'oj.
\newblock Robust pricing-hedging dualities in continuous time.
\newblock \emph{Finance and Stochastics}, 22\penalty0 (3):\penalty0 511--567,
  2018.

\bibitem[It{\^o} and Watanabe(1965)]{ito1965transformation}
Kiyosi It{\^o} and Shinzo Watanabe.
\newblock Transformation of {M}arkov processes by multiplicative functionals.
\newblock \emph{Annales de l'Institut Fourier}, 15\penalty0 (1):\penalty0
  13--30, 1965.

\bibitem[Karandikar(1995)]{Karandikar}
Rajeeva~L. Karandikar.
\newblock On pathwise stochastic integration.
\newblock \emph{Stochastic Processes and their Applications}, 57\penalty0
  (1):\penalty0 11--18, 1995.

\bibitem[Karatzas and Shreve(1998)]{KaratzasShreve1998}
Ioannis Karatzas and Steven~E. Shreve.
\newblock \emph{Methods of Mathematical Finance}, volume~39 of
  \emph{Probability Theory and Stochastic Modelling}.
\newblock Springer, 1998.

\bibitem[Knight(1921)]{Knight}
Frank~H. Knight.
\newblock \emph{Risk, Uncertainty, and Profit}.
\newblock Houghton Mifflin, Boston, 1921.

\bibitem[Le~Cam(1957)]{lecam1957convergence}
Lucien Le~Cam.
\newblock Convergence in distribution of stochastic processes.
\newblock \emph{University of California Publications in Statistics},
  2:\penalty0 207--236, 1957.

\bibitem[Li et~al.(2022)Li, Liu, and Rutkowski]{li2022vulnerable}
Libo Li, Ruyi Liu, and Marek Rutkowski.
\newblock Vulnerable european and american options in a market model with
  optional hazard process.
\newblock \emph{arXiv preprint arXiv:2212.12860}, 2022.

\bibitem[Merton(1973)]{Merton}
Robert~C. Merton.
\newblock Theory of rational option pricing.
\newblock \emph{The Bell Journal of Economics and Management Science},
  4\penalty0 (1):\penalty0 141--183, 1973.

\bibitem[Meyer(1978)]{MeyerBaxterChacon}
P.~A. Meyer.
\newblock Convergence faible et compacite des temps d'arret d'apres baxter et
  chacon.
\newblock In \emph{S{\'e}minaire de Probabilit{\'e}s XII}, pages 411--423.
  Springer Berlin Heidelberg, 1978.

\bibitem[Mykland(2003)]{Mykland}
Per~Aslak Mykland.
\newblock Financial options and statistical prediction intervals.
\newblock \emph{The Annals of Statistics}, 31\penalty0 (5):\penalty0
  1413--1438, 2003.

\bibitem[Myneni(1992)]{Myneni}
Ravi Myneni.
\newblock The pricing of the {A}merican option.
\newblock \emph{Ann. Appl. Probab.}, 2\penalty0 (1):\penalty0 1--23, 1992.

\bibitem[Neuberger(2007)]{Neuberger}
Anthony Neuberger.
\newblock Bounds on the {A}merican option.
\newblock \emph{Preprint, http://ssrn.com/abstract=966333}, 2007.

\bibitem[Neufeld and Nutz(2013)]{NeufeldNutz2013}
Ariel Neufeld and Marcel Nutz.
\newblock {Superreplication under volatility uncertainty for measurable
  claims}.
\newblock \emph{Electronic Journal of Probability}, 18:\penalty0 1--14, 2013.

\bibitem[Ob{\l}{\'o}j(2004)]{OblojSEP}
Jan Ob{\l}{\'o}j.
\newblock The {S}korokhod embedding problem and its offspring.
\newblock \emph{Probability Surveys}, 1:\penalty0 321--392, 2004.
\newblock \doi{10.1214/154957804100000060}.

\bibitem[Obłój and Wiesel(2021)]{OWUnified}
Jan Obłój and Johannes Wiesel.
\newblock A unified framework for robust modelling of financial markets in
  discrete time.
\newblock \emph{Finance and Stochastics}, 25\penalty0 (3):\penalty0 427--468,
  2021.

\bibitem[Peng(2019)]{Peng}
Shige Peng.
\newblock \emph{Nonlinear Expectations and Stochastic Calculus under
  Uncertainty: with Robust CLT and G-Brownian Motion}.
\newblock Springer Berlin Heidelberg, 2019.

\bibitem[Perkowski and Pr{\"o}mel(2015)]{Promeletal}
Nicolas Perkowski and David Pr{\"o}mel.
\newblock {Local times for typical price paths and pathwise Tanaka formulas}.
\newblock \emph{Electronic Journal of Probability}, 20:\penalty0 1--15, 2015.

\bibitem[Possama{\"i} et~al.(2013)Possama{\"i}, Royer, and
  Touzi]{PossamaiRoyerTouzi}
Dylan Possama{\"i}, Guillaume Royer, and Nizar Touzi.
\newblock {On the robust superhedging of measurable claims}.
\newblock \emph{Electronic Communications in Probability}, 18:\penalty0 1--13,
  2013.

\bibitem[Sentilles(1972)]{sentilles1972bounded}
F.~Dennis Sentilles.
\newblock Bounded continuous functions on a completely regular space.
\newblock \emph{Transactions of the American Mathematical Society},
  168:\penalty0 311--336, 1972.

\bibitem[Shmaya and Solan(2014)]{shmaya2014equivalence}
Eran Shmaya and Eilon Solan.
\newblock Equivalence between random stopping times in continuous time.
\newblock \emph{arXiv preprint arXiv:1403.7886}, 2014.

\bibitem[Soner et~al.(2013)Soner, Touzi, and Zhang]{SonerTouziZhang}
H.~Mete Soner, Nizar Touzi, and Jianfeng Zhang.
\newblock Dual formulation of second order target problems.
\newblock \emph{Annals of Applied Probability}, 23\penalty0 (1):\penalty0
  308--347, 2013.

\bibitem[Stewart(2012)]{IanStewart17}
Ian Stewart.
\newblock \emph{Seventeen equations that changed the world}.
\newblock Profile, London, 2012.

\bibitem[Tan et~al.(2013)Tan, Touzi, et~al.]{tan2013optimal}
Xiaolu Tan, Nizar Touzi, et~al.
\newblock Optimal transportation under controlled stochastic dynamics.
\newblock \emph{The annals of probability}, 41\penalty0 (5):\penalty0
  3201--3240, 2013.

\bibitem[Vovk(2012)]{VovkFS}
Vladimir Vovk.
\newblock Continuous-time trading and the emergence of probability.
\newblock \emph{Finance and Stochastics}, 16\penalty0 (4):\penalty0 561--609,
  2012.

\bibitem[Walley(1991)]{WalleyImpreciseProbabilities}
Peter Walley.
\newblock \emph{Statistical Reasoning with Imprecise Probabilities}.
\newblock Chapman \& Hall, 1991.

\end{thebibliography}

\end{document}